\documentclass[12pt]{article}
\usepackage{graphicx, xcolor}
\usepackage{caption}
\usepackage{algorithm}
\usepackage{newfloat}
\usepackage{listings}
\usepackage{natbib}
\usepackage{amsmath}
\usepackage{amsthm}
\usepackage{amssymb}
\usepackage[noend]{algorithmic}
\usepackage[capitalise]{cleveref}
\usepackage{comment}
\usepackage{tikz}
\usepackage{tikz-network}
\usetikzlibrary{positioning}

\usepackage{graphicx, xcolor}
\usepackage{algorithm}
\usepackage[noend]{algorithmic}
\usepackage[capitalise]{cleveref}
\usepackage{tikz}
\usepackage{tikz-network}

\newtheorem{theorem}{Theorem}
\newtheorem{corollary}[theorem]{Corollary}
\newtheorem{lemma}[theorem]{Lemma}

\newcommand{\Z}{\mathbb{Z}}

\newcounter{algpre}[algorithm]
\newcounter{algpost}[algorithm]

\newcommand{\PRECOND}{\refstepcounter{algpost}\refstepcounter{algpre}\REQUIRE}

\newcommand{\POSTCOND}{\refstepcounter{algpre}\refstepcounter{algpost}\ENSURE}

\crefname{algpre}{Precondition}{Preconditions}
\crefname{algpost}{Postcondition}{Postconditions}

\newcommand{\target}[1]{\ensuremath{\mathit{target}(#1)}}
\newcommand{\pebble}[1]{\ensuremath{\mathit{pebble}(#1)}}
\newcommand{\demand}[1]{\ensuremath{d(#1)}}
\newcommand{\balance}{\ensuremath{\mathit{balance\_subtrees}}}
\newcommand{\inject}{\ensuremath{\mathit{inject\_pebble}}}
\newcommand{\extract}{\ensuremath{\mathit{extract\_pebble}}}
\newcommand{\move}{\ensuremath{\mathit{move\_pebble}}}
\newcommand{\OPT}{\ensuremath{\mathit{OPT}}}
\newcommand{\directtraffic}
{\ensuremath{\mathit{direct\_traffic}}}
\newcommand{\sendagent}
{\ensuremath{\mathit{send\_agent\_from}}}

\begin{document}
\title{Optimal Unlabeled Pebble Motion on Trees and its Application to Multi-Agent Path Finding}
\author{Annalisa Calvi, Pierre Le Bodic\footnote{pierre.lebodic@monash.edu}, Samuel McGuire, Edward Lam}

\maketitle

\begin{abstract}
    Given a tree, a set of pebbles initially stationed at some nodes of the tree, and a set of target nodes, the Unlabeled Pebble Motion on Trees problem (UPMT) asks to find a plan to move the pebbles one-at-a-time from the starting nodes to the target nodes along the edges of the tree while minimizing the number of moves.
    This paper proposes the first optimal algorithm for UPMT that is asymptotically as fast as possible, as it runs in a time \emph{linear} in the size of the input (the tree) and the size of the output (the optimal plan). We extend this to solve unlabeled Multi-Agent Path Finding (MAPF) in trees, providing novel bounds on optimal makespan, sum of costs, and pebble motion plan length.
\end{abstract}

\section{Introduction}

The input of the Unlabeled Pebble Motion on Trees problem (UPMT) is a tree $T=(V, E)$, with $n = |V|$.
A set of $k$ pebbles is initially located at $k$ distinct nodes of the tree, called \emph{starting nodes}.
There are also $k$ nodes described as \emph{targets}, some of which may coincide with the starting nodes.
A \emph{move} is an action that moves a single pebble from its current node to an adjacent pebble-less node.
A \emph{plan} is a sequence of moves.
A \emph{feasible} plan moves every pebble to a target.
The \emph{length} of a plan is the number of moves it uses.
The length of an optimal UPMT plan is denoted \OPT{}.

In the \emph{labeled} version of the problem, every starting node and every target node is labeled with a specific pebble, and every target node must be reached by its designated pebble.
In \emph{unlabeled} (or \emph{anonymous}) Pebble Motion problems, a target can be reached by any pebble.

Pebble Motion is related and has been used \citep[see e.g.][]{kulich2019} to tackle Multi-Agent Path Finding (MAPF)~\citep{stern2019}, a problem in which pebbles/agents can move synchronously, and hence lends itself well to real-life problems such as automated warehouses. The length of a path taken by an agent in MAPF is the sum of its move actions and wait actions away from its target. Feasible plans in MAPF can be evaluated using the \emph{makespan} or \emph{sum of costs} metrics; makespan is the maximum path length of an agent, while sum of costs is the sum of the path length of all agents. As with the Pebble Motion problem, MAPF can be labeled or unlabeled. This paper discusses unlabeled MAPF in trees.

This paper proposes an optimal algorithm for UPMT with time complexity $O(n\log n + \OPT \log n)$, 
where the encoding size of the input is $O(n \log n)$, and the encoding size of the plan is $O(\OPT{} \log n)$\footnote{This result also appears in the conference paper \cite{lebodic2024}.}. We also propose an algorithm to solve unlabeled MAPF in trees suboptimally, with the same time complexity. This algorithm guarantees a feasible plan with makespan at most $n-k$ and sum of costs at most $k(n-k)$.

We also show that $\OPT{} \leq k(n-k)$.
Complementary to this worst-case result, we establish that, for random trees, the length of an optimal plan is actually much smaller.
Specifically, for a distribution $X_n$ of trees on $n$ nodes, we find that $\mathbb{E}(\OPT{}) \leq \sqrt{Dk(n-k)}$, where $D$ is the average distance between two nodes over trees in $X_n$. Notably, when $X_n$ is the uniform distribution on labeled trees, we have $D \sim \sqrt{\pi n/2}$~\citep{meirmoon1970}, so that $\mathbb{E}(\OPT{}) \leq \sqrt{k(n-k)\sqrt{\pi n/2}}$.

The paper is organized as follows.
%
First, we review related work.
%
Second, we derive a tight lower bound on the length of an optimal plan.
%
Third, we introduce an optimal algorithm using this lower bound.
Next, we extend these ideas to a suboptimal unlabeled MAPF algorithm. Then, we deduce bounds on \OPT{}, makespan, and sum of costs. We go on to examine experimental and theoretical average plan length and how it compares to our bounds.
Naturally, we end with conclusions.

\section{Related Work}
\label{sec:related}
\citet{kornhauser1984} propose a sub-optimal algorithm for UPMT with plans of length $O(n^3)$.
\citet{ardizzoni2023} provide a more readable description, as well as an improvement to path lengths in $O(k n c + n^2)$, where $c$ is the length of a \emph{corridor}, which is at worst $O(n)$.

\citet{auletta1999} provide an algorithm linear in $n$ to determine the feasibility of the labeled problem. In the feasible cases, they also provide an algorithm that runs in $O(n+\OPT_l)$, where $\OPT_l$ refers to the length of a plan for the labeled problem, which is in $O(k^2(n-k))$.

\citet{calinescu2008} study (un)labeled problems similar to the Pebble Motion problem on various classes of graphs, but define a move to be a consecutive sequence of edge moves made by a single pebble.
They also give an optimal algorithm for UPMT.
This algorithm first computes all shortest paths between pebbles and targets, and then solves a minimum-weight matching problem to find an assignment of pebbles to targets that minimizes the total distance between pebbles and targets.
The time complexity of solving a minimum matching in a complete bipartite graph with $2k$ nodes and edge weights of at most $n$ is $O(k^{2.5} \log n)$~\citep{gabow1989} or $O(k^3)$~\citep{edmonds1972}.

\citet{yulavelle2013} use network flow techniques to solve unlabeled MAPF in general graphs $(V, E)$ in $O(kVE)$ time, with plans of makespan at most $V + k - 1$.
\citet{okumura2023} provide the TSWAP algorithm for unlabeled MAPF, which can run online or offline on a general graph $G$. Although the plan output is not necessarily optimal in makespan or sum of costs, the makespan is bounded by $k\cdot \text{diam}(G)$ and the sum of costs is bounded by $k^2 \cdot\text{diam}(G)$.

\section{Lower Bound on the Optimal Plan Length}
\label{sec:lowerbound}

Pick an arbitrary root node $r \in V$.
Let $T_u$ denote the subtree rooted at a node $u \in V$.
We write $v \in T_u$ to mean that $v$ is a node of $T_u$.
We define $\pebble{u} \in \{0,1\}$ and $\target{u} \in \{0,1\}$ as the current number of pebbles and targets at a node $u \in V$, respectively.
Note that $\pebble{u}$ is updated as pebbles move, while $\target{u}$ is a fixed input.
We define the \emph{demand} $d:V\rightarrow \Z$, such that, for $u \in V$,
\begin{equation}
    \demand{u} = \sum_{v \in T_u} \target{v} - \sum_{v \in T_u} \pebble{v},
\label{eq:d_sum_def}
\end{equation}
as the number of targets minus the current number of pebbles in $T_u$. See Figure \ref{fig:DemandExample} for an example.
An equivalent recursive definition is
\begin{equation}
    \demand{u} = \target{u} - \pebble{u} + \sum_{v \text{ child of } u}\demand{v}.
\label{eq:d_rec_def}
\end{equation}
\begin{figure}
    \centering
     \begin{tikzpicture}
     \SetDistanceScale{2}
        \Vertex[label=0, x=0, y=4, color=green, shape = diamond]{A} 
        \Vertex[label=0, x=-1, y=3]{B}
        \Vertex[label=3, x=0, y=3, color=green, shape = diamond]{C}
        \Vertex[label=-2, x=1, y=3, color=green, shape = diamond]{D}
        \Vertex[label=1, x=-1, y=2, color=red, shape = rectangle]{E}
        \Vertex[label=3, x=0, y=2, color=red, shape = rectangle]{F}
        \Vertex[label=-1, x=1, y=2, color=green, shape = diamond]{G}
        \Vertex[label=1, x=-2, y=1, color=red, shape = rectangle]{H}
        \Vertex[label=1, x=-1, y=1, color=red, shape = rectangle]{I}
        \Vertex[label=-1, x=0, y=1, color=green, shape = diamond]{J}
        \Vertex[label=0, x=1, y=1, color=green, shape = diamond]{K}
        \Vertex[label=1, x=2, y=1]{L}
        \Vertex[label=0, x=-2.5, y=0, color=red, shape = rectangle]{M}
        \Vertex[label=0, x=-2.5, y=0, color=green, shape = diamond]{M2}
        \Vertex[label=1, x=-1.5, y=0, color=red, shape = rectangle]{N}
        \Vertex[label=-1, x=-0.5, y=0, color=green, shape = diamond]{O}
        \Vertex[label=1, x=0.5, y=0, color=red, shape = rectangle]{P}
        \Vertex[label=0, x=1.5, y=0]{Q}
        \Vertex[label=1, x=2.5, y=0, color=red, shape = rectangle]{R}
        \Edge(A)(B)
        \Edge(A)(C)
        \Edge(A)(D)
        \Edge(C)(E)
        \Edge(C)(F)
        \Edge(D)(G)
        \Edge(F)(H)
        \Edge(F)(I)
        \Edge(G)(J)
        \Edge(G)(K)
        \Edge(G)(L)
        \Edge(H)(M)
        \Edge(H)(N)
        \Edge(H)(O)
        \Edge(K)(P)
        \Edge(K)(Q)
        \Edge(L)(R)
    \end{tikzpicture}
    \caption{Example configuration of nodes labeled with demands. Green diamonds represent pebbles and red squares represent targets.}
    \label{fig:DemandExample}
\end{figure}
\begin{lemma}
The problem is solved if and only if $\demand{u}=0$ for every $u \in V$.
\label{thm:zero_demand}
\end{lemma}
\begin{proof}
$(\Rightarrow)$ Suppose the problem is solved, i.e. $\target{u} = \pebble{u}~\forall u \in V$, then \eqref{eq:d_sum_def} implies $d=0$.
$(\Leftarrow)$ By induction.
At every leaf $u$, \eqref{eq:d_rec_def} gives $0 = \target{u} - \pebble{u} + 0$.
Furthermore, for every non-leaf node $u$, suppose that $\demand{v}=0$ for each child $v$ of $u$. Again, \eqref{eq:d_rec_def} simplifies to $0 = \target{u} - \pebble{u} + 0$.
Therefore, at every node $u$, $\target{u} = \pebble{u}$.
\end{proof}

\begin{lemma}
\label{thm:lower_bound}
A feasible plan has length at least $\sum_{u \in V} |\demand{u}|$.
\end{lemma}
\begin{proof}
Consider a non-root node $v \in V \setminus \{r\}$ and its parent $u \in V$.
Suppose that $\demand{v} \geq 0$, i.e. there are $\demand{v}$ more targets than pebbles in the subtree $T_v$.
Any feasible plan must include at least $\demand{v}$ moves from $u$ through $v$, otherwise the subtree $T_v$ will have missing pebbles after the execution of the plan.
In the case where $T_v$ has extra pebbles, i.e. $\demand{v} \leq 0$,
there must be $-\demand{v}$ moves from $v$ through $u$.
Therefore, in a feasible plan, the number of moves on edge $uv$ is at least $|\demand{v}|$.
Finally, observe that $\demand{r} = 0$.
\end{proof}

\section{Optimal Unlabeled Pebble Motion on Trees}
\label{sec:algorithm}

We now use the demand $d$ to design an algorithm that makes no unnecessary moves.
To do so, only moves that decrease the lower bound (\cref
{thm:lower_bound}) are performed.
We present a recursive top-down algorithm that realizes this.


\subsection{Algorithm Description}
\cref{alg:balance_subtrees}, \balance{}, starts at a node $u$ with zero demand (\cref{pre:balance_zero_demand}).
The first part of \balance{}, Lines \ref{alg:balance_part_1_start}-\ref{alg:balance_part_1_end}, ensures that every child of $u$ has zero demand, by balancing the pebbles between the children.
A balance between children is achievable because $d(u)=0$.
For the same reason, \balance{} implicitly handles the case of $u$ being a target because ensuring all children are balanced means that a pebble is left on $u$ if and only if $u$ is a target.
Node $u$ interacts with the subtree of a child $v$ via the functions $\inject(v)$ and $\extract(v)$, not by direct moves between $u$ and $v$, as moving pebbles from one subtree to another can require moves inside these subtrees.
The second part of \balance{}, Lines \ref{alg:balance_part_2_start}-\ref{alg:balance_part_2_end}, recursively calls \balance{} on all children to ensure that all descendants have zero demand (\cref{post:balance_zero_demand}).

\cref{alg:inject_pebble}, \inject{}, is called on $v$ by its parent $u$ when a pebble needs to be moved from $u$ to $v$. It handles the case where $v$ already holds a pebble. If so, $v$ must first \inject{} its pebble to one of its children, before the pebble on $u$ is moved to $v$. Because $d(v)>0$ (\cref{pre:inject_need_pebbles}), this is guaranteed to succeed.

\cref{alg:extract_pebble}, \extract{}, moves a pebble from $v$ to its parent $u$. It handles the case where there is no pebble on $v$ by recursively extracting a pebble from within $T_v$. Thanks to \cref{pre:extract_extra_pebbles}, this is guaranteed to succeed.

\cref{alg:move_pebble}, \move{}, updates \pebble{\cdot} and \demand{\cdot}, and finally performs a move.

\begin{algorithm}
    \caption{$\balance(u)$}\label{alg:balance_subtrees}
    \begin{algorithmic}[1]
        \PRECOND A node $u \in V$
        \label{pre:balance_input_node}
        \PRECOND $\demand{u}=0$
        \label{pre:balance_zero_demand}
        \POSTCOND $\demand{v}=0$ for all $v \in T_u$
        \label{post:balance_zero_demand}
        \WHILE{a child $v$ of $u$ with $d(v) \neq 0$ exists}
        \label{alg:balance_part_1_start}\label{alg:balance_while_start}
            \IF{$\pebble{u} = 1$}\label{alg:balance_if}
                \STATE Pick a child $v$ of $u$ with $\demand{v} > 0$ \label{alg:balance_subtrees_pick_v_d_pos}
                \STATE $\inject(v)$
            \ELSE
                \STATE Pick a child $v$ of $u$ with $\demand{v} < 0$ \label{alg:balance_subtrees_pick_v_d_neg}
                \STATE $\extract(v)$
            \ENDIF
        \ENDWHILE
        \label{alg:balance_part_1_end}
        \FOR{$v$ child of $u$}
        \label{alg:balance_part_2_start}
        \label{alg:balance_start_for}
            \STATE $\balance(v)$
        \ENDFOR
        \label{alg:balance_part_2_end}
    \end{algorithmic}
\end{algorithm}

\begin{algorithm}
    \caption{$\inject(v)$}\label{alg:inject_pebble}
    \begin{algorithmic}[1]
        \PRECOND A node $v \in V \setminus \{r\}$ with parent $u$
        \label{pre:inject_v_not_r}
        \PRECOND $\pebble{u} = 1$
        \label{pre:inject_pebble_on_u}
        \PRECOND $\demand{v} > 0$
        \label{pre:inject_need_pebbles}
        \POSTCOND $\pebble{u} = 0$
        \label{post:inject_no_pebble_on_u}
        \POSTCOND $\pebble{v} = 1$
        \label{post:inject_pebble_on_v}
        \IF{$\pebble{v} = 1$}\label{alg:inject_if}
            \STATE Pick a child $w$ of $v$ such that $\demand{w}>0$\label{alg:inject_pick_child}
            \STATE $\inject(w)$\label{alg:inject_recursive_call}
        \ENDIF
        \STATE $\move(u,v)$\label{alg:inject_call_to_move}
    \end{algorithmic}
\end{algorithm}

\begin{algorithm}
    \caption{$\extract(v)$}
    \label{alg:extract_pebble}
    \begin{algorithmic}[1]
        \PRECOND A node $v \in V \setminus \{r\}$ with parent $u$
        \label{pre:extract_v_not_r}
        \PRECOND $\pebble{u} = 0$
        \label{pre:extract_no_pebble}
        \PRECOND $\demand{v} < 0$
        \label{pre:extract_extra_pebbles}
        \POSTCOND $\pebble{u} = 1$
        \POSTCOND $\pebble{v} = 0$
        \IF{$\pebble{v}=0$}
            \STATE Pick a child $w$ of $v$ such that $\demand{w}<0$
            \STATE $\extract(w)$
        \ENDIF
        \STATE $\move(v,u)$
    \end{algorithmic}
\end{algorithm}

\begin{algorithm}
    \caption{$\move(u,v)$}\label{alg:move_pebble}
    \begin{algorithmic}[1]
        \PRECOND A node $u \in V$ adjacent to a node $v \in V$\label{pre:move_u_v_input}
        \PRECOND $\pebble{u} = 1$
        \label{pre:move_pebble_on_u}
        \PRECOND $\pebble{v} = 0$
        \label{pre:move_no_pebble_on_v}
        \PRECOND if $v$ is a child of $u$, then $\demand{v}>0$
        \label{pre:move_v_child}
        \PRECOND if $u$ is a child of $v$, then $\demand{u}<0$
        \label{pre:move_u_child}
        \POSTCOND $\pebble{u} = 0$
        \label{post:move_no_pebble_on_u}
        \POSTCOND $\pebble{v} = 1$
        \label{post:move_pebble_on_v}
        \STATE $\pebble{u}  \gets  0$
        \STATE $\pebble{v}  \gets  1$
        \IF{$v$ is a child of $u$}
            \STATE $\demand{v} \gets \demand{v}-1$
            \label{alg:move_d_v_decrement}
        \ELSE 
            \STATE $\demand{u} \gets \demand{u}+1$
            \label{alg:move_d_u_increment}
        \ENDIF
        \STATE Output move $(u,v)$
    \end{algorithmic}
\end{algorithm}

\subsection{Algorithm Correctness}
We prove that the algorithm terminates after having ensured that all pebbles are located at a target.

\subsubsection{Correctness of \move{}}
\begin{lemma}
A call to \move{} correctly updates $d$.
\end{lemma}
\begin{proof}
Assuming that $d$ is correct when \move{} is called, we prove it is still correct after the call.
A move from $u$ to $v$ can only change $d(u)$ and $d(v)$.
If $v$ is a child of $u$, moving a pebble from $u$ to $v$ does not change the number of pebbles in $T_u$. But the number of pebbles in $T_v$ increases by 1, therefore the demand $\demand{v}$ decreases by 1.
Similarly, if $v$ is the parent of $u$,  moving a pebble from $u$ to $v$ increases the demand $\demand{u}$ by 1.
\end{proof}

\begin{lemma}\label{thm:move_decrements_d}
A call to \move{} decrements $\sum_{u \in V} |\demand{u}|$ by 1.
\end{lemma}
\begin{proof}
If $v$ is a child of $u$, then by \cref{pre:move_v_child}, $\demand{v} > 0$, and it is decremented on Line \ref{alg:move_d_v_decrement}.
Otherwise, by \cref{pre:move_u_child}, $\demand{u} < 0$, and it is incremented on Line \ref{alg:move_d_u_increment}.
The value of $d$ at other nodes are unchanged.
\end{proof}

\begin{lemma}\label{thm:move_postcond}
\cref{post:move_no_pebble_on_u,post:move_pebble_on_v} are satisfied.
\end{lemma}

\subsubsection{Correctness of \inject{}}

\begin{lemma}
\Crefrange{pre:inject_v_not_r}{pre:inject_need_pebbles} are satisfied when \inject{} calls itself recursively on Line \ref{alg:inject_recursive_call}.
\end{lemma}
\begin{proof}
Note that $\pebble{v}=1$ on Line \ref{alg:inject_recursive_call}.
\Cref{pre:inject_v_not_r} and \cref{pre:inject_pebble_on_u} are clearly satisfied.
To show \cref{pre:inject_need_pebbles}, we check that there must be a child $w$ of $v$ such that $\demand{w}>0$ on Line \ref{alg:inject_pick_child}. Evaluating \eqref{eq:d_rec_def} for $v$ gives:
\begin{align*}
& \demand{v} = \target{v} - \pebble{v} + \sum_{w \text{ child of } v}\demand{w} \\
\implies & 0 < \demand{v} \leq \max(0,1) - 1 + \sum_{w \text{ child of } v}\demand{w} \\
\implies & 0 < \sum_{w \text{ child of } v}\demand{w}.
\end{align*}
\end{proof}

\begin{lemma}\label{thm:move_precond_inject}
\Crefrange{pre:move_u_v_input}{pre:move_u_child} are satisfied when \move{} is called by \inject{} on Line \ref{alg:inject_call_to_move}.
\end{lemma}
\begin{proof}
\Cref{pre:move_u_v_input}: node $u$ is the parent of $v$, therefore they are adjacent.

\Cref{pre:move_pebble_on_u}: \cref{pre:inject_pebble_on_u} remains true until \move{} is called on Line \ref{alg:inject_call_to_move}, as only pebbles within $T_v$ may have moved before Line \ref{alg:inject_call_to_move}.

\Cref{pre:move_no_pebble_on_v}: if $\pebble{v}=1$ on Line \ref{alg:inject_if}, then \inject{} is called on Line \ref{alg:inject_recursive_call}, which, by \cref{post:inject_no_pebble_on_u}, ensures $\pebble{v}=0$ on Line \ref{alg:inject_call_to_move}.

\Cref{pre:move_v_child}: \cref{pre:inject_need_pebbles} remains true until Line \ref{alg:inject_call_to_move}, since the number of pebbles in $T_v$ does not change before Line \ref{alg:inject_call_to_move}.


\end{proof}

\begin{lemma}
\Cref{post:inject_no_pebble_on_u,post:inject_pebble_on_v} are satisfied.
\end{lemma}
\begin{proof}
Direct consequence of Line \ref{alg:inject_call_to_move} and  \cref{thm:move_postcond}.
\end{proof}

\subsubsection{Correctness of \extract{}}
The results and proofs are similar to those of \inject{}.

\subsubsection{Correctness of \balance{}}

\begin{lemma}\label{thm:inject_precond_balance}
\Crefrange{pre:inject_v_not_r}{pre:inject_need_pebbles} are satisfied when \inject{} is called by \balance{}.
\end{lemma}
\begin{proof}
\Cref{pre:inject_v_not_r} is satisfied, as $v$ has a parent $u$.
\Cref{pre:inject_pebble_on_u} is satisfied, since $\pebble{u}=1$ on Line \ref{alg:balance_if}.
To show that \cref{pre:inject_need_pebbles} is satisfied, we show that there is indeed a child $v$ of $u$ with $\demand{v} > 0$.
Suppose there is not. Then, applying \eqref{eq:d_rec_def} at $u$,
\begin{align*}
\demand{u} &= \target{u} - \pebble{u} + \sum_{v \text{ child of } u}\demand{v} \\
&\leq \max(0,1) - 1 - \sum_{v \text{ child of } u}|\demand{v}|\\
&\leq - \sum_{v \text{ child of } u}|\demand{v}|\\
& < 0,
\end{align*}
where the last inequality holds because the while condition on Line \ref{alg:balance_while_start} is true.
Since $\demand{u}=0$ (by \cref{pre:balance_zero_demand}, and no pebble moving out of $T_u$), this is a contradiction.
\end{proof}

\begin{lemma}
\Crefrange{pre:extract_v_not_r}{pre:extract_extra_pebbles} are satisfied when \extract{} is called by \balance{}.
\end{lemma}
\begin{proof}
The proof is similar to that of \cref{thm:inject_precond_balance}.
\end{proof}

\begin{lemma}
\label{thm:balance_preconditions}
\Cref{pre:balance_input_node,pre:balance_zero_demand} are satisfied for all children $v$ of $u$ when \balance{} calls itself recursively.
\end{lemma}
\begin{proof}
The while condition on Line \ref{alg:balance_while_start} is false on Line \ref{alg:balance_start_for}, therefore $d(v)=0$ for every child $v$ of $u$.
 No pebble moves between $u$ and any of its children on or after Line \ref{alg:balance_start_for}, hence this remains true for all recursive calls.
\end{proof}

\begin{lemma}\label{thm:balance_postcond}
\cref{post:balance_zero_demand} is satisfied.
\end{lemma}
\begin{proof}
We first observe that, after each iteration of the while loop, at least one call to \move{} is made via \inject{} or \extract{}, therefore, by \cref{thm:move_decrements_d}, $\sum_{v \text{ child of } u} |\demand{v}|$ is decremented by at least 1.
Hence, Line \ref{alg:balance_start_for} is ultimately reached.
Since \balance{} is then called on all descendants of $u$ via recursion, and since \cref{pre:balance_zero_demand} is satisfied for all these calls (\cref{thm:balance_preconditions}), all nodes $v \in T_u$ satisfy $\demand{v}=0$ at the end of the call.
\end{proof}

\begin{theorem}\label{thm:balance_is_feasible}
Algorithm \balance{} produces a feasible plan for UPMT.
\end{theorem}
\begin{proof}
Direct from \cref{thm:zero_demand} and \cref{thm:balance_postcond}.
\end{proof}

\subsection{Algorithm Optimality}
We show that the plan output by \balance{} has minimum size, as it meets the bound given in \cref{thm:lower_bound}.

\begin{theorem}
\label{thm:optimality}
Algorithm \balance{} outputs a plan of optimal length $\OPT = \sum_{u \in V} |\demand{u}|$.
\end{theorem}
\begin{proof}
Each pebble move is the result of exactly one call to \move{}.
\cref{thm:move_decrements_d} ensures that \move{} can be called at most $\sum_{u \in V} |\demand{u}|$ times.
Thus, the plan produced by \balance{} has the length at most $\OPT$.
Since the plan is feasible (\cref{thm:balance_is_feasible}), its length is $\OPT$ (\cref{thm:lower_bound}).
\end{proof}

The results below further characterize optimal plans.

\begin{corollary}
$\OPT \leq k (n-1)$
\end{corollary}
\begin{proof}
Direct from \cref{thm:optimality}, the fact that $|d(u)| \leq k$ for all $u \in V$, and $d(r)=0$.
\end{proof}
Corollary \ref{thm:OPT-upper-k-n-k} will further improve this bound to $OPT \leq k(n - k)$.

\begin{theorem}
In an optimal plan, every pebble reaches its target via a shortest path.
\end{theorem}
\begin{proof}
\cref{pre:move_v_child,pre:move_u_child} ensure that a move between two nodes can occur in at most one direction throughout the algorithm.
Therefore, since a feasible plan is returned, all paths are simple, i.e. do not cycle.
It remains only to recall that there is a unique simple path between two nodes in a tree~\citep{diestel2017}, and thus it is shortest.

While this proof only shows that optimal plans produced by \balance{} have this property, note that every move in these plans is necessary to return all demands to zero. Then, if a pebble doubles back over an edge, it makes an unnecessary move, in addition to all the necessary moves. Then this plan would not be optimal.
\end{proof}

\subsection{Algorithm Complexity}

We suppose that a reference (e.g. index or pointer) to a node uses $O(\log n)$ bits, and thus reading or writing a reference to a node takes $O(\log n)$ time.
Therefore, the size of the encoding of an n-ary tree on $n$ nodes is $n\log n$.
Furthermore, the encoding of a plan is the number of moves times two node indices, one for each endpoint, so its size is $O(\OPT \log n)$ if it is optimal.

We prove that the runtime of this algorithm is asymptotically optimal, as it runs in a time linear in the encoding size  of the input plus the encoding size of the plan, i.e. $O(n\log n + \OPT \log n)$.

In the proofs that follow, observe that, if, for the analysis, we chose a computational model where all numbers could be encoded in constant size, all $\log$ terms would disappear.
The runtime would then simplify to $O(n + \OPT)$, which, in this model, would also be the encoding size of the input and the encoding size of the output.

\begin{lemma}\label{thm:runtime_of_d}
Function $d$ can be computed in $O(n \log n)$ time.
\end{lemma}
\begin{proof}
Note that $d$ holds numbers of encoding size at most $O(\log k)$.
The base case of the recursive definition \eqref{eq:d_rec_def} of $d$ occurs at the leaves, thus $d$ can be computed by a postorder traversal of $T$.
At each node $u$, we access all of its $n_u$ children, which takes $O(n_u\log n)$ time, and add up the $n_u+2$ terms of \eqref{eq:d_rec_def}, which takes $O(n_u\log k)$ time.
Since $k \leq n$, this is $O(n_u\log n)$ per node.
Since there are $\sum_{u \in V} n_u = n-1$ children in $T$, the runtime is in $O(n\log n)$.
\end{proof}

\begin{lemma}
\move{} runs in $O(\log n)$.
\end{lemma}
\begin{proof}
This is simply the time to output $u$ and $v$, access $u$, $v$ and related data, as well as the addition in $O(\log k)$ time to update $d$.
\end{proof}

\begin{lemma}
The total number of calls to \inject{} and \extract{} is \OPT{}.
\end{lemma}
\begin{proof}
Both \inject{} and \extract{} (whether it is a recursive call or not) call \move{} exactly once. \Cref{thm:optimality} shows that this happens \OPT{} times.
\end{proof}

\begin{lemma}
\label{thm:runtime_of_picking_child}
At a node $u$, picking a child $v$ with $\demand{v}>0$ or $\demand{v}<0$ takes $O(\log n)$ time.
\end{lemma}
\begin{proof}
When creating the tree, store each child $v$ of $u$ in one of three lists at $u$ according to whether $\demand{v}>0$, $\demand{v}=0$ or $\demand{v}<0$. Return the first node of each list as required, which takes $O(\log n)$. If, after a move, $\demand{v}$ is now 0, pop $v$ out of its list and place it in the $d=0$ list, which is in $O(1)$.

\end{proof}

\begin{corollary}
\label{thm:runtime_of_subroutines}
The total runtime for all calls to \inject{} and \extract{} is in $O(\OPT \log n)$.
\end{corollary}

\begin{lemma}
\label{thm:runtime_balance_iterations}
Across all calls, the total number of while (resp. for) loop iterations of \balance{} is at most \OPT{} (resp. $n-1$).
\end{lemma}
\begin{proof}
Each while loop iteration leads to at least one pebble move, of which there is at most \OPT{} across all calls to \balance{}.
\end{proof}

\begin{theorem}
\label{thm:overall_runtime_UPMT}
The runtime of \balance{} is in $O(n\log n + \OPT \log n)$.
\end{theorem}
\begin{proof}
Reading the input and computing $d$ (\cref{thm:runtime_of_d}) takes $O(n\log n)$ time.
All calls to functions \move{}, \inject{} and \extract{} together take $O(\OPT \log n)$ time (\cref{thm:runtime_of_subroutines}).
The total number of iterations in all calls to \balance{} is $O(\OPT+n)$ (\cref{thm:runtime_balance_iterations}).
\Cref{thm:runtime_of_picking_child} shows that each iteration takes $O(\log n)$ time.
\end{proof}

\section{Bounded Suboptimal Unlabeled Multi-Agent Path Finding on Trees}
\label{sec:mapf-algorithm}
In this section we propose an algorithm to sup-optimally solve the MAPF problem on trees. The algorithm takes a similar top-down approach to \balance{} to find an unlabeled MAPF solution. Agents make the same moves as in the previous algorithm, informed by the demand function.

\subsection{Algorithm Description}
\Cref{alg:process-st} uses similar principles to \Cref{alg:balance_subtrees}. As in \Cref{alg:balance_subtrees}, each agent traverses edges according to the demand function, but now wait actions must also be considered.

We consider the tree to be ordered so that children with negative demand are positioned to the left of their parent, while those with positive demand are to the right. The algorithm performs a single in-order traversal of the tree. At each vertex, agents may be received from the parent, and then agents are pulled from each node on the left by a sequence of recursive calls, then all these agents are sent to the right. To ensure incoming agents do not clash, each new incoming agent is delayed until all previously processed agents have passed before being allowed to move to the vertex. We ensure all wait actions are performed at the start, meaning each agent's path consists of an initial sequence of wait actions followed by uninterrupted movement. 

To illustrate this, consider a node $u$ with two negative-demand children, $v$ and $w$, where $v$ is to the left of $w$. Agents are only moved from $w$ to $u$ after all agents moving from $v$ to $u$ have been sent away from $u$. Suppose further that $v$ is the left-most child of $u$, and $u$ is a positive-demand child of $x$. In this case agents may only move from $v$ to $u$ once all agents moving from $x$ to $u$ have been moved through $u$. In our ordering of the tree, $x$ is left of $v$ which is left of $w$ which is left of $u$. So all agents are moving from left to right, and are giving way to agents coming from further left.

Before execution, each node $u$ in $T$ is assigned a demand $d(u)$, as in the UPMT algorithm. This demand is updated throughout the algorithm. The initial demand of node $u$ is denoted $d_I(u)$. 

Additionally, each vertex $u$ has an associated list $l(u)$ (initially empty) and positive integer value $s(u)$ (initially 0). When $l(u)$ is empty, we define $\max l(u) = 0$ and $\min l(u) = \infty$ for convenience. 
The list $l(u)$ primarily stores the times at which an agent passes through $u$. 
However, if $u$ is a target node and the algorithm has terminated, then $\max l(u)$ is the timestep at which the final agent arrives at $u$. 
In this case, $u$ is occupied by an agent from timestep $\max l(u)$ onwards. 
We will see that algorithms \ref{alg:process-st} and \ref{alg:process-peb} ensure $l(u)$ is always sorted in increasing order, with no duplicates.

The integer $s(u)$ represents the earliest timestep at which we will allow an agent to leave $u$. 
If $d_I(u) \leq 0$, then $s(u) = 0$ for the duration of the algorithm. 
Otherwise, this value is updated when \directtraffic{} is run on $parent(u)$, to avoid collisions with other agents passing through $parent(u)$. 

 Figure \ref{fig:DemandExample2} shows the final values of $l(u)$ and $s(u)$ at each vertex $u$ of an example problem.

\Cref{alg:process-st}, \directtraffic{}, takes in a node $u$. 
It outputs all moves $(v, w, t)$ in the plan with $v \in T_u$. 
A call to \directtraffic{}$(\mathit{r})$ outputs a complete plan.

This algorithm performs an ``in-order'' traversal of the tree. 
A call to $\directtraffic{}(u)$ first processes the children of $u$ with negative demand
(lines \ref{line:neg-child-start}--\ref{line:neg-child-end}), 
then processes all moves leaving $u$ 
(lines \ref{line:self-start}--\ref{line:self-end}), 
then processes the remaining children of $u$ 
(lines \ref{line:pos-child-start}--\ref{line:pos-child-end}).

\Cref{alg:process-peb}, \sendagent{}, is a subroutine of \Cref{alg:process-st}. 
It takes in a node $u$ and a time $t$, and outputs a move $(u, v, t)$ that starts at node $u$ and occurs at time $t$. 
The receiving node $v$ is chosen arbitrarily from all $v$ for which the move $(u, v)$ conforms to the demands.
We add $t+1$ to the list $l(v)$. 
If $u$ is a target node, then the last agent arriving at $u$ will be left at $u$; 
that is,
\sendagent{}$(u, \max l(u))$ will do nothing.

\begin{algorithm}
\caption{\directtraffic{}$(u)$}\label{alg:process-st}
\begin{algorithmic}[1]
\PRECOND $d(u) \leq 0$  \label{cond: init-demand}
\PRECOND If $d(u) < 0$ then $l(u)$ is empty  \label{cond: lu-empty}
\PRECOND If $d(u) < 0$ then $s(u) \ge \max l(parent(u))$ and $s(u) + 1 \ge s(parent(u))$ \label{cond: su-max-l}
\PRECOND $l(u)$ is sorted with no duplicates, $\min l(u) \ge 1$ \label{cond: lu-sorted}
\PRECOND If $d(u) = 0$ then $s(u) = 0$ \label{cond: su-0}
\POSTCOND $d(u) = 0$ \label{cond: post-demand}
\POSTCOND $d(v) = 0$ for each $v$ child of $u$ \label{cond: post-ch-demand}
\POSTCOND $l(u)$ is sorted with no duplicates \label{cond: post-sorted}
\POSTCOND $s(u) \leq \min l(u)$ \label{cond: post-min}
\IF {$agent(u) = 1$}
    \STATE insert $s(u)$ at start of $l(u)$ \label{line:start-agent}
\ENDIF
\FOR {$v$ child of $u$ with $d(v) < 0$} \label{line:neg-child-start}
\STATE $s(v) \gets \max(s(u)-1, \max l(u), 0)$  \label{line: s'}
\STATE $\directtraffic{}(v)$ \label{line:neg-child-end}
\ENDFOR
\FOR {time $t$ in $l(u)$} \label{line:self-start}
\STATE $\sendagent(u, t)$ \label{line:self-end}
\ENDFOR 
\STATE mark $u$ as processed \label{line:mark-processed}
\FOR {$v$ child of $u$ not yet processed} \label{line:pos-child-start}
\STATE $\directtraffic{}(v)$ \label{line:pos-child-end}
\ENDFOR 
\end{algorithmic}
\end{algorithm}

\begin{algorithm}
\caption{$\sendagent(u, t)$}\label{alg:process-peb}
\begin{algorithmic}[1]
\IF {$d(u) < 0$}
\STATE append $t+1$ to $l(parent(u))$ \label{line:up-t+1}
\STATE output move $(u, parent(u), t)$
\STATE $d(u) \gets d(u) + 1$ \label{line: inc-demand}
\ELSIF {a child $v$ of $u$ with $d(v) > 0$ exists}
\STATE pick a child $v$ of $u$ with $d(v) > 0$ \label{line: pick-child}
\STATE append $t+1$ to $l(v)$ \label{line: down-t+1}
\STATE output move $(u, v, t)$
\STATE $d(v) \gets d(v) - 1$ \label{line: red-demand}
\ENDIF
\end{algorithmic}
\end{algorithm}

\begin{figure}
    \centering
     \begin{tikzpicture}
     \SetDistanceScale{2}
        \Vertex[IdAsLabel, x=0, y=2, color=red, shape = rectangle]{D} 
        \Vertex[IdAsLabel, x=-1.5, y=1]{B}
        \Vertex[IdAsLabel, x=-0.5, y=1, color=green, shape = diamond]{C}
        \Vertex[IdAsLabel, x=1.5, y=1, color=red, shape = rectangle]{F}
        \Vertex[IdAsLabel, x=-2, y=0, color=green, shape = diamond]{A}
        \Vertex[IdAsLabel, x=1, y=0, color=green, shape = diamond]{E}
        \Vertex[IdAsLabel, x=2, y=0, color=red, shape = rectangle]{G}
        \Text[position=above, x=0, y=2, distance=0.3, width = 1.2cm]{$s: 0$ $l:2, 3$}
        \Text[position=left, x=-1.5, y=1, distance=0.3, width = 1.2cm]{$s: 0$ $l:1$}
        \Text[position=left, x=-2, y=0, distance=0.3, width = 1.2cm]{$s: 0$ $l:0$}
        \Text[position=below, x=-0.5, y=1, distance=0.3, width = 1.2cm]{$s: 2$ $l:2$}
        \Text[position=right, x=1.5, y=1, distance=0.4, width = 1.2cm]{$s: 0$ $l:3, 4$}
        \Text[position=left, x=1, y=0, distance=0.3, width = 1.2cm]{$s: 3$ $l:3$}
        \Text[position=right, x=2, y=0, distance=0.4, width = 1.2cm]{$s: 0$ $l:4$}
        \Edge(A)(B)
        \Edge(B)(D)
        \Edge(C)(D)
        \Edge(D)(F)
        \Edge(E)(F)
        \Edge(F)(G)
    \end{tikzpicture}
    \caption{Example configuration of nodes, labeled $A - G$ in order of processing. The final values of $s(u)$ and $l(u)$ are shown. The agent at $A$ travels the path $A - B - D - F - G$ with no initial wait. The agent at $C$ waits at timesteps $0$ and $1$, giving way to the agent at $A$. The agent at $E$ waits at timesteps $0$, $1$ and $2$, giving way to the agent at $A$.}
    \label{fig:DemandExample2}
\end{figure}
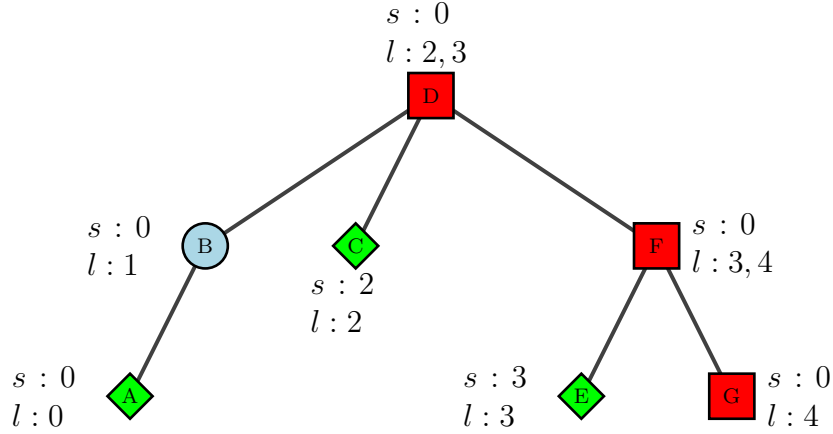

\subsection{Algorithm Correctness}

We define the following predicates on nodes $u$ in $T$.
    \begin{itemize}
        \item $pre(u)$: All preconditions are satisfied when \directtraffic{} is called on $u$.
        \item $neg\_post(u)$: For each child $v$ of $u$ with $d_I(v) < 0$, the postconditions of $\directtraffic{}(v)$ are satisfied.
        \item $get\_up(u)$: For each child $v$ of $u$ with $d_I(v) < 0$, $u$ is updated with $|d_I(v)|$ new items $t+1$ with $t \in l(v)$ when $v$ is processed.
        \item $get\_down(u)$: If $d_I(u) > 0$, then $|l(u)| = d_I(u)$ when \directtraffic{} is called on $u$.
        \item $post(u)$: All postconditions are satisfied when $u$ is processed.
        \item $send\_up(u)$: If $d_I(u) < 0$, then $parent(u)$ is updated with $|d_I(u)|$ new items $t+1$ with $t \in l(u)$ when $u$ is processed.
        \item $send\_down(u)$: For each $w$ child of $u$ with $d_I(w) > 0$, $l(w)$ is updated with $|d(w)|$ new items $t+1$, with $t \in l(u)$, in ascending order with no duplicates, when $u$ is processed.
        \item $leave(u)$: If $target(u) = 0$, an agent passes through $u$ at each time $t \in l(u)$. If $target(u) = 1$, the last agent in $l(u)$ stays at $u$.
    \end{itemize}

\begin{lemma} \label{lem: get-up-neg-post}
    If $post(v)$ and $send\_up(v)$ hold for all children $v$ of $u$ with $d_I(v) < 0$, then $neg\_post(u)$ and $get\_up(u)$ are true.
\end{lemma}
\begin{proof}
    Follows immediately from definitions.
\end{proof}
\begin{lemma} \label{lem:neg-precond}
    If $d_I(u) \leq 0$, then $pre(u)$ is true.
\end{lemma}
\begin{proof}

    Case $d_I(u) < 0$: We call \directtraffic{} on $u$ on line \ref{line:neg-child-end} of the call to $\directtraffic{}(parent(u))$.
    At this stage, $d(u) = d_I(u) < 0$, so precondition \ref{cond: init-demand} is satisfied. We have not yet processed any agents going through $parent(u)$, nor any agents in $T_u$, so $l(u)$ is empty, satisfying \ref{cond: lu-empty}. On line \ref{line: s'} of $\directtraffic{}(parent(u))$, we ensure $s(u) \ge \max l(parent(u))$ and $s(u) + 1 \ge s(parent(u))$, satisfying precondition \ref{cond: su-max-l}. The other preconditions are vacuously true.

     Case $d_I(u) = 0$: As every move must decrease the overall sum of demands, we know no agents have been moved to this node when we begin processing. That is, $d(u) = d_I(u) = 0$, satisfying precondition \ref{cond: init-demand}, and $l(u)$ is empty, satisfying precondition \ref{cond: lu-sorted}. As $d(u)$ is not negative, $s(u)$ is not updated on line \ref{line: s'} of $\directtraffic{}(parent(u))$; therefore $s(u) = 0$, satisfying precondition \ref{cond: su-0}. Preconditions \ref{cond: lu-empty} and \ref{cond: su-max-l} are vacuously true.

     As all preconditions are satisfied, $pre(u)$ is true.
\end{proof}

\begin{lemma} \label{lem: process-neg}
If $pre(u)$ and $neg\_post(u)$ are satisfied, then $\directtraffic{}(u)$ meets postconditions \ref{cond: post-sorted}, \ref{cond: post-min}, and, for $v$ child of $u$ with $d_I(v) < 0$, \ref{cond: post-ch-demand}.
\end{lemma}
\begin{proof}
     Suppose $pre(u)$ and $neg\_post(u)$ are true. After line \ref{line:start-agent} of $\directtraffic{}(u)$, $l(u)$ may or may not contain $s(u)$.
    
    Case $d(u) < 0$: $l(u)$ was originally empty by precondition \ref{cond: lu-empty}, so it is now sorted with no duplicates and $s(u) \leq \min l(u)$. 
    
    Case $d(u) = 0$: $s(u) = 0$ by precondition \ref{cond: su-0} and all entries in $l(u)$ are greater than 0 by \ref{cond: lu-sorted}, so inserting $s(u)$ at the start of $l(u)$ maintains that $l(u)$ is sorted with no duplicates. Cases end here.
    
    Now we iterate over the children $v$ of $u$ with $d(v) < 0$. By $neg\_post(u)$ and lemma \ref{lem:neg-precond}, $\directtraffic{}(v)$ satisfies all pre and post conditions.
    
    As $\directtraffic{}(v)$ meets precondition \ref{cond: su-max-l}, $\max l(u) \leq s(v)$; then with postcondition \ref{cond: post-min}, we have $\max l(u) \leq s(v) \leq \min l(v)$. 
    We know $l(v)$ is sorted with no duplicates (by postcondition \ref{cond: post-sorted}), and $t + 1 > \max l(u)$ for each $t \in l(v)$ (since $\max l(u) \leq \min l(v)$). Thus, appending $t+1$ to $l(u)$ for $t$ in $l(v)$, as is done in $\sendagent{}(v, t)$, maintains that $l(u)$ is sorted with no duplicates.
    
    This is true on each iteration of line \ref{line:neg-child-start} of $\directtraffic{}(u)$, so when $l(u)$ is finalised, we have satisfied postcondition \ref{cond: post-sorted}. Additionally, again using \ref{cond: su-max-l} and \ref{cond: post-min} for $\directtraffic{}(v)$, we have $s(u) \leq s(v) + 1 \leq \min l(v) + 1$. Therefore, after appending $t+1$ to $l(u)$ for $t \in l(v)$, we maintain $s(u) \leq \min l(u)$ and satisfy postcondition \ref{cond: post-min}. As $\directtraffic{}(v)$ meets postcondition \ref{cond: post-demand}, each child $v$ of $u$ which has been processed now has $d(v) = 0$.
\end{proof}

\begin{lemma} \label{lem: send-up-neg-demand}
    If $d_I(u) \leq 0$ and $get\_up(u)$ is true, then $send\_up(u)$ is true, postcondition \ref{cond: post-demand} of $\directtraffic{}(u)$ is satisfied, and there are
    \(
    target(u) + \sum_{\substack{w \text{ child of } u \\ d_I(w) > 0}} |d_I(w)| 
    \)
    agents that are not sent upwards.
\end{lemma}
\begin{proof}
    Suppose $d_I(u) \leq 0$ and $get\_up(u)$ is true. 
    Then, when executing line \ref{line:self-start} of $\directtraffic{}(u)$, for each child $v$ of $u$ with $d_I(v) < 0$, we have appended $|d_I(v)|$ values to $l(u)$, as well as potentially including $s(u)$ if $agent(u) = 1$. Therefore,
    \[|l(u)| = agent(u) + \sum_{\substack{v \text{ child of } u \\ d_I(v) < 0}} |d_I(v)|,\]
    so that
    \[
    |l(u)| = |d_I(u)| + target(u) + \sum_{\substack{w \text{ child of } u \\ d_I(w) > 0}} |d_I(w)| 
    \]
    by definition of $d_I(u)$. This gives the number of times \sendagent{} is called on line \ref{line:self-end}. 
    
    We can see that in the first $|d_I(u)|$ calls to \sendagent{}, an agent is sent up to $parent(u)$. That is, $t+1$ is appended to $l(parent(u))$ for $|d(u)|$ values of $t$ in $l(u)$. So $send\_up(u)$ is true. Each of these times, we increment $d(u)$ on line \ref{line: inc-demand}, so that after these moves we have $d(u) = 0$, satisfying postcondition \ref{cond: post-demand}. 

    Finally, as $|d_I(u)|$ of $|l(u)|$ agents have been sent upwards, there are $target(u) + \sum_{\substack{w \text{ child of } u \\ d_I(w) > 0}} |d_I(w)|$ agents \emph{not} sent upwards.
\end{proof}

\begin{lemma} \label{lem: send-down-pos-demand}
If $d_I(u) > 0$ and $get\_up(u)$ and $get\_down(u)$ are true, then there are
    \(
    target(u) + \sum_{\substack{w \text{ child of } u \\ d_I(w) > 0}} |d_I(w)| 
    \)
    agents that are not sent upwards.
\end{lemma}
\begin{proof}
    Suppose $d_I(u) > 0$ and $get\_down(u)$ is true.
    By $get\_down(u)$, $|l(u)|$ is initially equal to $|d(u)|$. After processing $agent(u)$ and each child $v$ of $u$ with negative demand, by $get\_up(u)$ we have
    \[|l(u)| = |d(u)| + agent(u) + \sum_{\substack{v \text{ child of } u \\ d(v) < 0}} |d(v)|,\]
    so that
    \[
    |l(u)| = target(u) + \sum_{\substack{w \text{ child of } u \\ d(w) > 0}} |d(w)|
    \]
    by definition of $d_I(u)$. As $d(u) > 0$, no agents are sent upwards (see $\sendagent{}(u, t)$), so indeed we have \(
    target(u) + \sum_{\substack{w \text{ child of } u \\ d_I(w) > 0}} |d_I(w)| 
    \)
    agents that are not sent upwards.
\end{proof}

\begin{lemma} \label{lem: ABCDEF}
    If $pre(u)$, $neg\_post(u)$, $get\_up(u)$ and $get\_down(u)$ are true, then $post(u)$, $send\_up(u)$, $send\_down(u)$ and $leave(u)$ are true.
\end{lemma}
\begin{proof}
    Suppose $pre(u)$, $neg\_post(u)$, $get\_up(u)$ and $get\_down(u)$ are true. 

    Note that, if $d_I(u) \leq 0$, then $d(u) = 0$ when $\directtraffic{}(u)$ is called, so that postcondition \ref{cond: post-demand} is trivially satisfied. Additionally, $send\_up(u)$ is vacuously true. Lemma \ref{lem: send-up-neg-demand} handles the case $d_I(u) > 0$, so $send\_up(u)$ and postcondition \ref{cond: post-demand} are always satisfied.

    By lemmas \ref{lem: send-up-neg-demand} and \ref{lem: send-down-pos-demand}, at any node $u$ there are $
    target(u) + \sum_{\substack{w \text{ child of } u \\ d_I(w) > 0}} |d_I(w)| 
    $
    agents passing through $u$ that are not sent upwards. This results in sending $|d(w)|$ agents down into each child $w$ of $u$ with $d(w) > 0$, decrementing $d(w)$ on line \ref{line: red-demand} of $\sendagent{}(u, t)$ each time $w$ receives an agent. This leaves exactly $target(u)$ agents at node $u$, satisfying $leave(u)$. 
    
    Consequently, all children of $w$ of $u$ with $d_I(w) > 0$ have zero demand after $u$ is processed, satisfying half of postcondition \ref{cond: post-ch-demand}. Incorporating lemma \ref{lem: process-neg}, postcondition \ref{cond: post-ch-demand} is fully satisfied, as are postconditions \ref{cond: post-sorted} and \ref{cond: post-min}.
    Therefore, all postconditions are satisfied after \directtraffic{} is called on $u$, so $post(u)$ is true.

    As postcondition \ref{cond: post-sorted} is satisfied, by line \ref{line:self-start} of $\directtraffic{}(u)$, $l(u)$ is sorted with no duplicate values. When processing agents in lines \ref{line:self-start} - \ref{line:self-end}, the times $t \in l(u)$ are processed in increasing order. Consequently, for each $w$ child of $u$ with $d(w) > 0$, we will append $|d_I(w)|$ times to $l(w)$ in increasing order with no duplicates. Therefore $send\_down(u)$ is true.
\end{proof}

\begin{lemma} \label{lem:neg-postcond}
    The statements $post(u)$, $send\_up(u)$, $send\_down(u)$ and $leave(u)$ are true for all nodes $u$ with $d_I(u) < 0$.
\end{lemma}
\begin{proof}
    By induction. Assume $d_I(u) < 0$.
    
    Base case: Suppose $u$ has no children with negative demand. By lemma \ref{lem:neg-precond}, all preconditions are satisfied, so $pre(u)$ is true. As there are no children with negative demand, $neg\_post(u)$ and $get\_up(u)$ are vacuously true. As $d(u) < 0$, $get\_down(u)$ is vacuously true also. Applying lemma \ref{lem: ABCDEF} shows $post(u)$, $send\_up(u)$, $send\_down(u)$ and $leave(u)$ are true.

    Induction: Suppose that for each $v$ child of $u$ with $d(v) < 0$, $post(v)$, $send\_up(v)$, $send\_down(v)$ and $leave(v)$ are true.

    We have $pre(u)$ by lemma \ref{lem:neg-precond}. By lemma \ref{lem: get-up-neg-post}, $neg\_post(u)$ and $get\_up(u)$ are true.
    
     Again, $get\_down(u)$ is vacuously true. Then, by lemma \ref{lem: ABCDEF}, $post(u)$, $send\_up(u)$, $send\_down(u)$ and $leave(u)$ are true.
\end{proof}
\begin{corollary} \label{cor: BC-allu}
    The statements $neg\_post(u)$ and $get\_up(u)$ are true for all nodes $u$.
\end{corollary}
\begin{proof}
    Follows from the above lemma and lemma \ref{lem: get-up-neg-post}.
\end{proof}

\begin{lemma}\label{lem:pos-postcond}
    The statements $post(u)$, $send\_up(u)$, $send\_down(u)$ and $leave(u)$ are true for all nodes $u$ with $d_I(u) \ge 0$.
\end{lemma} 
\begin{proof}
    By induction. Assume $d_I(u) \ge 0$.

    Base case: Suppose $u$ satisfies $d_I(u) = 0$. At least one such $u$ exists, as $d_I(r)=0$, where $r$ is the root node. Then $pre(u)$ is true by lemma \ref{lem:neg-precond}, $neg\_post(u)$ and $get\_up(u)$ are true by corollary \ref{cor: BC-allu}, and $get\_down(u)$ is vacuously true. Therefore $post(u)$, $send\_up(u)$, $send\_down(u)$ and $leave(u)$ are all true, by lemma \ref{lem: ABCDEF}.

    Induction: Suppose $post(parent(u))$, $send\_up(parent(u))$, $send\_down(parent(u))$ and $leave(parent(u))$ are true. By corollary \ref{cor: BC-allu}, $neg\_post(u)$ and $get\_up(u)$ are true. 
    
    As $parent(u)$ satisfies postcondition \ref{cond: post-ch-demand}, $d(u) = 0$ when $u$ is processed, satisfying precondition \ref{cond: init-demand}. Preconditions \ref{cond: lu-empty} and \ref{cond: su-max-l} are vacuously true. We can see that $s(u)$ is not updated before calling \directtraffic{}$(u)$ on line \ref{line:pos-child-end} of $\directtraffic{}(parent(u))$, so $s(u) = 0$, satisfying precondition \ref{cond: su-0}. Finally, by $send\_down(parent(u))$, $l(u)$ is initially sorted in increasing order with no duplicates, and contains $|d(u)|$ items, satisfying $get\_down(u)$. As all times $t$ in $l(parent(u))$ are non-negative, all times $t+1$ in $l(u)$ satisfy $t+1 \ge 1$. Therefore, precondition \ref{cond: lu-sorted} is satisfied, so $pre(u)$ is true.

    By lemma \ref{lem: ABCDEF}, we conclude that $post(u)$, $send\_up(u)$, $send\_down(u)$ and $leave(u)$ are true. 
\end{proof}

\begin{theorem} \label{thm:feasible}
    Algorithm \directtraffic{} produces a feasible plan for unlabeled MAPF on trees.
\end{theorem}
\begin{proof}
     By lemmas \ref{lem:neg-postcond} and \ref{lem:pos-postcond}, all nodes $u$ in $T$ satisfy $post(u)$, and $leave(u)$. That is, after running $\directtraffic{}(r)$, all demands are 0, and $l(u)$ is sorted with no duplicates for all $u$. 
     
     As all demands are 0, we conclude that all agents are at a target when our plan concludes, although this may not be reflected by $agent(u)$ and $target(u)$, which were not updated. 
     
     By $leave(u)$, the last value in $l(u)$ may represent an agent staying at $u$; otherwise, for each $t \in l(u)$, the agent arrives at timestep $t$ and leaves at timestep $t + 1$. By postcondition \ref{cond: post-sorted}, this list contains no duplicate values. Therefore, no two agents are arriving at $u$ at the same time; as the agents leave immediately, this means no agent moves to a node which is occupied by another agent. Finally, as every edge move decreases the sum of demands, edges are only traversed in one direction. Therefore, no two agents will collide along an edge. 

     As all agents have reached a target by the end of the plan, and there are no collisions, algorithm $\directtraffic{}(r)$ outputs a feasible plan for unlabeled MAPF.
\end{proof}

\subsection{Algorithm Suboptimality}
Recall that makespan refers to the maximum time taken for an agent to reach its target, whereas sum of costs is the sum of the time taken for each agent to reach its target.

\begin{theorem}
    The algorithm does not always give optimal solutions for the makespan or sum of costs objectives. 
\end{theorem}
\begin{proof}
    In Figure \ref{fig:AlgorithmSubOptimalExample}, the optimal solution is given by the agent on node $E$ taking the path, $E-B-A-D$, and for $C$ to take $C-A-D-F$, giving a makespan and sum of costs of 3 and 6 respectively. However, \Cref{alg:process-st} would output the sequence of moves for $E$ as $E-B-A-D-F$, and for $C$, $C-C-C-A-D$. This gives a makespan and sum of costs of 4 and 8 respectively. An infinite family of such counterexamples can be found by subdividing the edges $AC$ and $BE$ an equal number of times. Hence, the algorithm is suboptimal in makespan and sum-of-costs. 
    \begin{figure}[ht]
    \centering
     \begin{tikzpicture}
        \Vertex[IdAsLabel]{A}
        \Vertex[IdAsLabel, x=-1, y=-1]{B}
        \Vertex[IdAsLabel, x=0, y=-1, color=green, shape = diamond]{C}
        \Vertex[IdAsLabel, x=1, y=-1, color=red, shape = rectangle]{D}
        \Vertex[IdAsLabel, x=-1, y=-2, color=green, shape = diamond]{E}
        \Vertex[IdAsLabel, x=1, y=-2, color=red, shape = rectangle]{F}
        \Edge(A)(D)
        \Edge(A)(B)
        \Edge(A)(C)
        \Edge(B)(E)
        \Edge(D)(F)
    \end{tikzpicture}
    \caption{}
    \label{fig:AlgorithmSubOptimalExample}
\end{figure}
\end{proof}

\begin{theorem}
    An optimal makespan or sum of costs solution may require edges to be traversed in both directions.
\end{theorem}
\begin{proof}
    If edges are only allowed to be traversed in one direction, then the optimal solution in Figure \ref{fig:MakespanSOCCounterexample} will have the agent at $F$ go to target $O$. This gives a makespan of 6, and a sum of cost of 20. If edges may be traversed in both directions, then the agent at $B$ will finish at $O$ instead, with the edge $AB$ being traversed by agents $B$ and $F$ in opposite directions. This gives a makespan of 5, and sum of costs of 19. Therefore an optimal solution may require edges to be traversed in different directions.
    \begin{figure}
    \centering
     \begin{tikzpicture}
        \Vertex[IdAsLabel]{A}
        \Vertex[IdAsLabel, x=-2, y=-1, color=green, shape = diamond]{B}
        \Vertex[IdAsLabel, x=-3, y=-2]{E}
        \Vertex[IdAsLabel, x=-5, y=-3, color=red, shape = rectangle]{H}
        \Vertex[IdAsLabel, x=-4, y=-3, color=red, shape = rectangle]{I}
        \Vertex[IdAsLabel, x=-3, y=-3, color=red, shape = rectangle]{J}
        \Vertex[IdAsLabel, x=-2, y=-3, color=red, shape = rectangle]{K}
        \Vertex[IdAsLabel, x=-1, y=-3, color=green, shape = diamond]{L}
        \Vertex[IdAsLabel, x=-1, y=-4, color=green, shape = diamond]{N}
        \Vertex[IdAsLabel, x=-1, y=-5, color=green, shape = diamond]{P}
        \Vertex[IdAsLabel, y=-1]{C}
        \Vertex[IdAsLabel, y=-2, color=green, shape = diamond]{F}
        \Vertex[IdAsLabel, x=2, y=-1]{D}
        \Vertex[IdAsLabel, x=2, y=-2]{G}
        \Vertex[IdAsLabel, x=2, y=-3]{M}
        \Vertex[IdAsLabel, x=2, y=-4, color=red, shape = rectangle]{O}

        \Edge(A)(C)
        \Edge(C)(F)
        \Edge(A)(B)
        \Edge(B)(E)
        \Edge(E)(L)
        \Edge(L)(N)
        \Edge(N)(P)
        \Edge(A)(D)
        \Edge(D)(G)
        \Edge(G)(M)
        \Edge(M)(O)
        \Edge(E)(H)
        \Edge(E)(I)
        \Edge(E)(J)
        \Edge(E)(K)
    \end{tikzpicture}
    \caption{}
    \label{fig:MakespanSOCCounterexample}
\end{figure}
\end{proof}

\subsection{Algorithm Complexity}
\Cref{alg:process-peb} runs in time asymptotically equivalent to that of reading the input and writing the output. Moreover, the output format is minimal, recording only moves along an edge and their corresponding times, rather than tracking each agent's position at every timestep.

\begin{lemma}
    \Cref{alg:process-peb} runs in $O(\log n)$.
\end{lemma}
\begin{proof}
    This is the time to output $u$, $v$ and $t$, each bounded above by $n$; as well as $O(\log k)$ time to update $d$. Line \ref{line: pick-child} can be executed in $O(\log n)$ time by lemma \ref{thm:runtime_of_picking_child}.
\end{proof}

\begin{lemma}
    The total number of calls to \sendagent{} is $OPT + k$.
\end{lemma}
\begin{proof}
    Almost every time \sendagent{} is called, the sum of demands is decremented. The only exception is when $t$ is the last time an agent ``passes through'' $u$, and $target(u) = 1$; this occurs once for each target, making $k$ times in total. So the total number of calls is $OPT + k$.
\end{proof}

\begin{theorem} \label{thm:time}
    The runtime of \directtraffic{} is in $O(n \log n + OPT \log n)$.
\end{theorem}
\begin{proof}
    Reading the input and computing $d$ takes $O(n\log n)$ time (lemma \ref{thm:runtime_of_d}). Algorithm \directtraffic{} is called exactly once on each node. Line \ref{line: s'} takes $O(\log n)$, as $s(u) - 1$ and $\max l(u)$ are bounded above by $n$ (see Theorem \ref{thm:makespan-upper-bound}); this is called once for each node with negative demand, taking $O(n \log n)$ time in total. Finally, \sendagent{} is called $OPT + k$ times, with all calls together taking $O(OPT\log n + k\log n)$ time. 

    Therefore \directtraffic{} runs in $O(n \log n + OPT \log n)$ time.
\end{proof}
\subsection{Subsequent bounds on optimal solutions}
Let the nodes be ordered $v_1, \ldots, v_n$ where $v_p$ is marked as processed before $v_{p+1}$ in \Cref{alg:process-st}. We can write $v_i < v_j$ whenever $i < j$. Then let $t_p$ be the total number of targets located at nodes $v_1, \ldots, v_{p-1}$. Recursively, we could define $t_1 = 0$ and $t_{p+1} = t_p + target(v_p)$.

After a node $u$ is marked as processed on line \ref{line:mark-processed} of $\directtraffic{}(u)$, we next call $\directtraffic{}(r)$ on on some child $r$ of $u$ not yet processed. If such a child does not exist, we resume processing $parent(u)$ at line \ref{line:neg-child-start}. Let $r(u)$ be this child $r$ if it exists, else $r(u) := parent(u)$. We can see that $r(u)$ is marked as processed after $u$, that is, $r(u) > u$. 

Let $ch(u)$ denote the first child of $u$ with negative demand that has not yet been processed. This is only well defined when such a child exists. We can see that $ch(u)$ is processed before $u$; we write $ch(u) < u$. As it is, by definition, the next child to be processed, it is processed before any remaining siblings. As $ch(u) < u$ for all $u$, we have
\[
ch \circ \ldots \circ ch(u) < \ldots < ch \circ ch(u) < ch(u) < u
\]
so long as $ch \circ \ldots \circ ch(u)$ exists. Let $gc(u) := ch \circ \ldots \circ ch(u)$, using the longest possible composition of $ch$ that allows this to be well-defined. If $u$ does not have any remaining children with negative demand, define $gc(u) := u$. 

\begin{lemma} \label{lem: next-v}
    For all $1 \leq p < n$, $v_{p+1} = gc(r(v_p))$.
\end{lemma}
\begin{proof}
    We saw that, after processing $v_p$, we resume processing $r(v_p)$. As $gc(r(v_p))$ has no children with negative demand, it will be processed before any of its children. Then, tracing \Cref{alg:process-st}, we can see that $gc(r(v_p))$ will be the first node processed when we resume processing $r(v_p)$. Therefore $gc(r(v_p)) = v_{p+1}$. This is a similar process to finding the successor of a node in a binary search tree.
\end{proof}

\begin{lemma} \label{lem: preserve-max}
    If $l(u)$ is nonempty or $d_I(u) = 0$ then $\max l(gc(u)) \leq \max l(u)$ when $gc(u)$ is processed.
\end{lemma}
\begin{proof}
    Case $d_I(u) = 0$: $s(u) = 0$ by precondition \ref{cond: su-0}. If $l(u)$ is empty, then we have not yet processed any children of $u$ with negative demand. As $\max l(u) = 0$ and $s(u) = 0$, on the first iteration of line \ref{line:neg-child-start} we assign $s(ch(u)) = 0$. Recursing down, we continue to set $s(ch \circ \ldots \circ ch(u)) = 0$, so that $s(gc(u)) = 0$ also. As $gc(u)$ has no negative demand children, either $agent(gc(u)) = 1$ and $l(gc(u)) = \{0\}$ or $agent(gc(u)) = 0$ and $l(gc(u)) = \{\}$. Either way, we write $\max l(gc(u)) = 0 = \max l(u)$.

    Case $l(u)$ is nonempty: we have $\max l(u) \ge \min l(u) \ge s(u)$. Therefore, when recursing into $ch(u)$ we assign $s(ch(u)) = \max l(u)$. Now, $l(ch(u))$ is empty by precondition \ref{cond: lu-empty}, so $s(ch\circ ch(u)) := \max (s(ch(u))-1, 0)$ on line \ref{line: s'}. Therefore $s(ch\circ ch(u)) \leq s(ch(u)) \leq \max l(u)$, and continuing by induction we get $s(gc(u)) \leq \max l(u)$. As $gc(u)$ has no children with negative demand, either $agent(gc(u)) = 0$ and $\max l(gc(u)) = 0 \leq \max l(u)$ or $agent(gc(u)) = 1$ and $\max l(gc(u)) = s(gc(u)) \leq \max l(u)$. 

\end{proof}

\begin{theorem}
    For all $1 \leq p \leq n$, we have $\max l(v_p) \leq p - 1 - t_p$.
\end{theorem}
\begin{proof}
The proof is by induction on $p$. 

Base case: We know $\max l(v_1) = 0$ because we have not moved any agents yet, so there cannot be any other agents passing through $v_1$. So $\max l(v_1) = 0 = 1 - 1 - t_1$, as $t_1 = 0$.

Induction step: Suppose $\max l(v_p) \leq p - 1 - t_p$ for some $1 \leq p < n$. If $target(v_p) = 0$, then at worst we are sending the last agent in $l(v_p)$ into $r(v_p)$, so that \[\max l(r(v_p)) \leq \max l(v_p) + 1 \leq p-t_p = p-t_p - target(v_p).\] However, if $target(v_p) = 1$, then the last agent in $l(v_p)$ stays at $v_p$, so does not transition to $r(v_p)$. In this case, an agent can be leaving $v_p$ no later than time $\max l(v_p) - 1$. We then get \[\max l(r(v_p)) \leq \max l(v_p) - 1 + 1 \leq p-t_p - 1 = p-t_p - target(v_p).\] 

Either $v_p$ sends an agent to $r(v_p)$, or $d(r(v_p)) = 0$ (potentially both if $r(v_p) = parent(v_p)$). If an agent is indeed sent, this is processed with $v_p$, so $l(r(v_p))$ is updated before we process $gc(r(v_p))$. Therefore we can invoke lemmas \ref{lem: next-v} and \ref{lem: preserve-max}, so that, after $v_{p+1}$ is processed, 
\begin{align*}
\max l(v_{p+1}) &= \max l(gc(r(v_p)))\\
&\leq \max l(r(v_p)) \\
&\leq p-t_p - target(v_p)\\
&= (p+1) - 1 -t_{p+1}
\end{align*}
using the recursive definition of $t_p$. So the induction holds.
\end{proof}

\begin{lemma} \label{lem:makespan}
    The makespan $M$ is equal to $\max_{u \in T} \max l(u)$.
\end{lemma}
\begin{proof}
    Recall that $\max l(u)$ is the latest time that an agent arrives at $u$. Therefore, taking the maximum over all nodes $u$ gives the makespan.
\end{proof}

\begin{theorem} \label{thm:makespan-upper-bound}
    Algorithm $process\_subtree$ gives a plan with makespan $M \leq n - k$.
\end{theorem}
\begin{proof}
    For all $1 \leq p < n$, we have that $p - 1 - t_p \leq (p + 1) - 1 - t_{p+1}$, so that $\max_{1\leq p\leq n} p - 1 - t_p = n - 1 - t_n$. Therefore, $\max_{u\in T} \max l(u) \leq n - 1 - t_n$. As $t_n$ is the number of targets in $v_1, \ldots, v_{n-1}$, we have $t_n \geq k - 1$. Finally, by lemma \ref{lem:makespan},
    \[
    M = \max_{u\in T} \max l(u) \leq n - t_n - 1 \leq n-k.
    \]
\end{proof}
\begin{corollary} \label{thm:sum-of-costs-upper-bound}
    Algorithm $process\_subtree$ gives a plan with sum of costs $S \leq k(n-k)$.
\end{corollary}
\begin{proof}
    The makespan $M$ gives the longest time it takes for any agent to reach its target, which is equivalent to the number of actions for one agent in the sum of costs objective. There are $k$ agents, each with a number of actions less than or equal to $n-k$, and so the total number of actions $S\leq k(n-k)$.
\end{proof}
\begin{corollary} \label{thm:OPT-upper-k-n-k}
    The plan length $OPT$ of UPMT is bounded by $OPT \leq k(n-k)$.
\end{corollary}
\begin{proof}
    The sum of costs $S$ consists of $OPT$ moving actions plus a non-negative number of wait actions. So $OPT \leq S \leq k(n-k)$.
\end{proof}

\begin{theorem} \label{thm:bounds-tight}
    The upper bounds in Theorem \ref{thm:makespan-upper-bound} and Corollaries \ref{thm:sum-of-costs-upper-bound} and \ref{thm:OPT-upper-k-n-k} are tight for all $n, k$.
\end{theorem}
\begin{proof}
    The configuration in Figure \ref{fig:TightUpperExample} has eight nodes and three agents. The agents are denoted by green diamonds and the targets are red rectangles. The MAPF solution will have a makespan of four, so $M = n-k$ so that the bound in Theorem \ref{thm:makespan-upper-bound} is tight. The sum of costs is 12, so that $S = k(n-k)$ and the bound in Corollary \ref{thm:sum-of-costs-upper-bound} is tight. As there are no wait actions, $S = OPT$ and the bound in Corollary \ref{thm:OPT-upper-k-n-k} is tight. Although this example uses $n=8$ and $k=3$, a similar construction shows the bound is tight for all $n, k$.

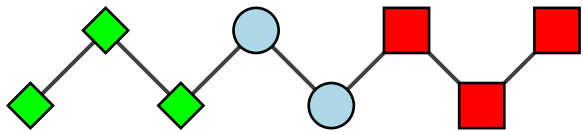
\begin{figure}
    \centering
     \begin{tikzpicture}
        \Vertex[x=0, y=0, color=green, shape = diamond]{A} 
        \Vertex[x=1, y=1, color=green, shape = diamond]{B}
        \Vertex[x=2, y=0, color=green, shape = diamond]{C}
        \Vertex[x=3, y=1]{D}
        \Vertex[x=4, y=0]{E}
        \Vertex[x=5, y=1, color=red, shape = rectangle]{F} 
        \Vertex[x=6, y=0, color=red, shape = rectangle]{G}
        \Vertex[x=7, y=1, color=red, shape = rectangle]{H}
        \Edge(A)(B)
        \Edge(B)(C)
        \Edge(C)(D)
        \Edge(D)(E)
        \Edge(E)(F)
        \Edge(F)(G)
        \Edge(G)(H)
    \end{tikzpicture}
    \caption{Example configuration where the bounds in Theorem \ref{thm:makespan-upper-bound} and corollaries \ref{thm:sum-of-costs-upper-bound} and \ref{thm:OPT-upper-k-n-k} are tight.}
    \label{fig:TightUpperExample}
\end{figure}
\end{proof}

\section{Bounds on average length of an optimal UPMT plan}
Corollary \ref{thm:OPT-upper-k-n-k} states that $OPT \leq k(n-k)$ for all trees and all configurations of agents and targets. Theorem \ref{thm:bounds-tight} states that this bound is tight for all choices of $n$ and $k$. Here, we assess how accurately this bound estimates the optimal plan length. Figure \ref{fig:n-vs-OPT} compares experimental data of $OPT$ with $k(n-k)$ for $1 \leq n \leq 10^6$, while Figure \ref{fig:k-vs-OPT} makes this comparison for $n = 10^6$ and $1 \leq k \leq n-1$. 
\begin{figure}[ht]
    \centering
    \includegraphics[width=\textwidth]{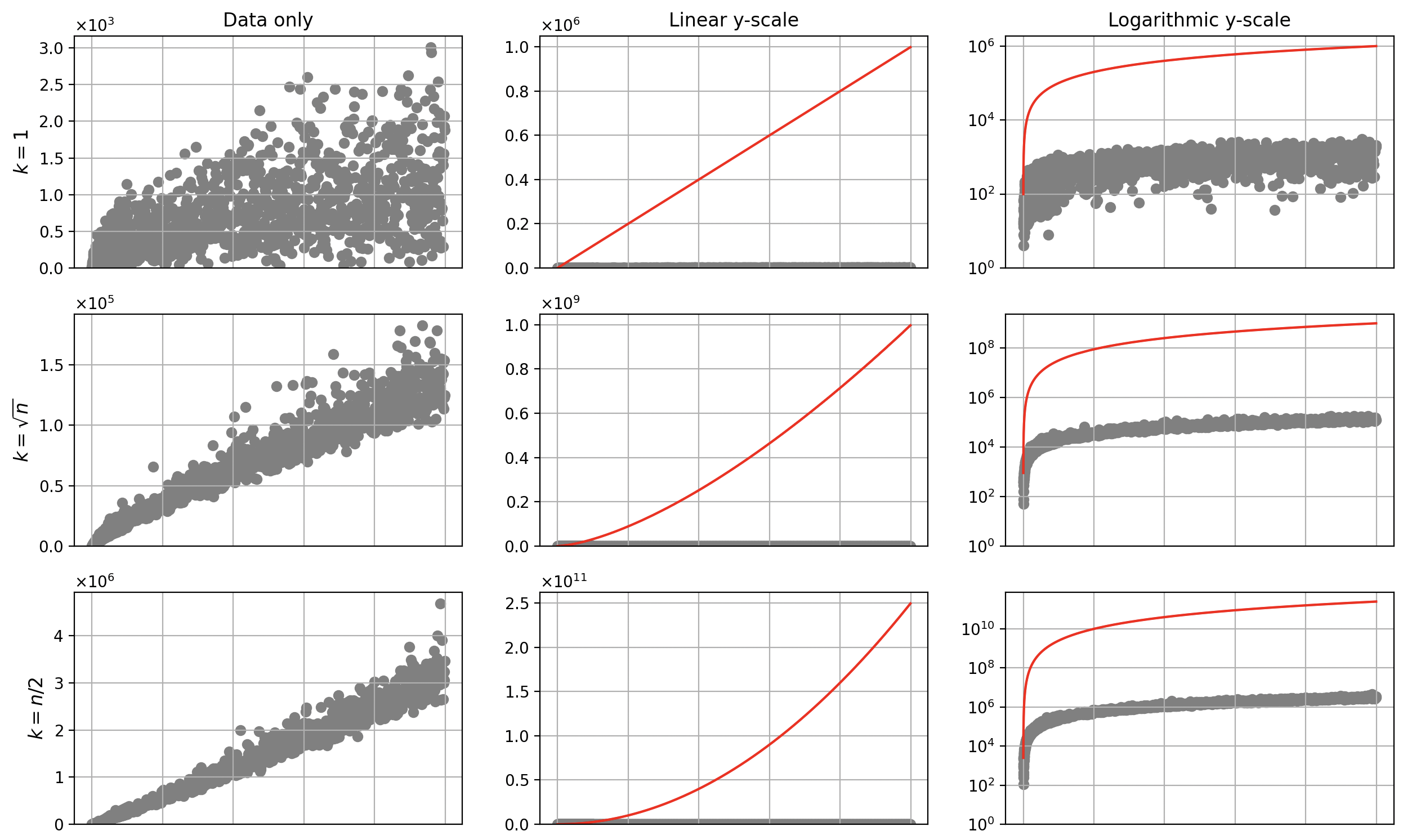}
    \caption{Comparison of $OPT$ (grey dots) and $k(n-k)$ (red line) in random labeled trees with up to $10^6$ nodes.}
    \label{fig:n-vs-OPT}
\end{figure}

\begin{figure}[ht]
    \centering
    \includegraphics[width=\textwidth]{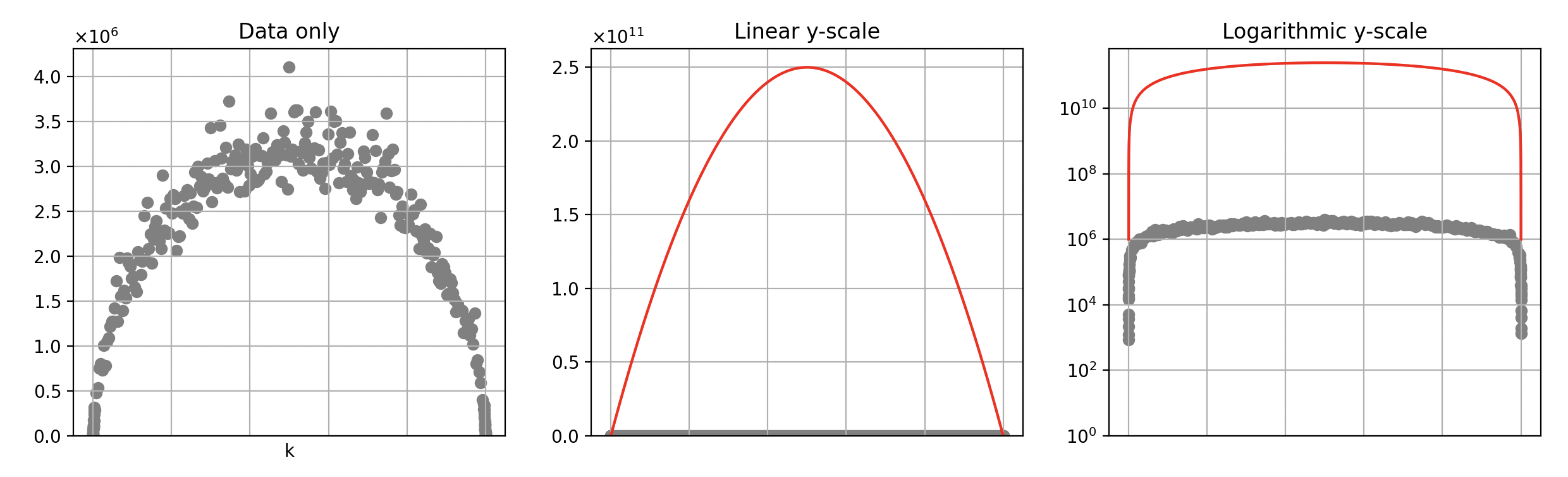}
    \caption{Comparison of $OPT$ (grey dots) and $k(n-k)$ (red line) over $k$ in random labeled trees with $10^6$ nodes.}
    \label{fig:k-vs-OPT}
\end{figure}

This data was obtained by picking a random labeled tree on $n$ nodes, a random sample of $k$ nodes for the initial agent positions, and a (potentially overlapping) random sample of $k$ nodes for the target positions. 

We can see that for these randomly chosen trees, the observed length of an optimal plan is far less than our upper bound $k(n-k)$. The graphs with logarithmic $y$-scale in particular suggest that we might be some \emph{power} off, ie. on average, $OPT = O((k(n-k))^c)$ for some $c < 1$. So far, this is only speculation. The relation $\OPT = \sum_{u \in V} |\demand{u}|$ enables a much more robust estimation of the optimal plan length.

Let $X_n$ be a distribution of rooted trees on $n$ nodes, with the root node labeled $n$ and the non-root nodes labeled $1, \ldots, n-1$. Suppose now that a tree $T$ is selected at random from the distribution $X_n$, and that two (possibly equal) nodes $u, v$ are selected uniformly at random. Let $D(X_n)$ denote the expected value of the distance between $u$ and $v$ in $T$, taken over all $u, v, T$. Then $D$ is \emph{not} a random variable, although its value depends on the distribution $X_n$.

\begin{lemma} \label{lem:D-value} Given a distribution $X_n$ of trees, the expected distance $D$ between two nodes is
    \[\frac{D}{n-1} = \sum_{T \in X_n} \text{Pr}(X_n = T)\sum_{u = 1}^{n-1}\frac{1}{n-1} 2 \frac{|T_u|}{n}\frac{n -|T_u|}{n}.\]
\end{lemma}
\begin{proof}
    Given a fixed tree $T$ and non-root node $u$, the probability that two randomly chosen nodes are on opposite sides of the edge $(u, parent(u))$ is $\frac{|T_u|}{n}\frac{n - |T_u|}{n} + \frac{n-|T_u|}{n}\frac{|T_u|}{n}$. The average distance between two nodes is equal to the expected number of nodes $u$ for which the two nodes are on opposite sides of this edge, so that
    \[
    D = \sum_{T \in X_n} \text{Pr}(X_n = T)\sum_{u = 1}^{n-1} 2 \frac{|T_u|}{n}\frac{n -|T_u|}{n}.
    \]
    The lemma follows.
\end{proof}

Now, let $\mathbb{E}(OPT(X_n, k))$ denote the average optimal plan length over all configurations made by choosing a tree $T$ from $X_n$, then placing agents on a random set $A$ of $k$ nodes, and placing targets on a random set $B$ of $k$ nodes. Placing an agent and a target at the same node is allowed in this construction.

\begin{theorem} \label{thm:e-opt-bound}
     For a distribution $X_n$ of trees, the expected plan length is bounded by $\mathbb{E}(OPT(X_n, k)) \leq \sqrt{Dk(n-k)}$.
\end{theorem}
\begin{proof}

We can express $\mathbb{E}(OPT(X_n, k))$ as
\begin{equation} \label{eq:E-opt-initial}
\sum_{T \in X_n} \text{Pr}(X_n = T)\sum_{u = 1}^{n-1}\sum_{\substack{A \subseteq [n] \\ |A| = k}}\frac{1}{{n\choose {k}}}\sum_{\substack{B \subseteq [n] \\ |B| = k}}\frac{1}{{n\choose{k}}}||B \cap T_u| - |A \cap T_u||,
\end{equation}
where we sum over all trees $T$, non-root nodes $u$, sets $A$ of agent positions and sets $B$ of target positions. Note that, for a given construction $T$, $u$, $A$ and $B$, the expression $||B \cap T_u| - |A \cap T_u||$ is equal to $|\demand{u}|$.

Let $\text{Ag}$ denote the random variable associated with $|A \cap T_u|$, for fixed $T_u$ and randomly selected $A$. Similarly, let $\text{Ta}$ denote the random variable associated with $|B \cap T_u|$. We can see that $\text{Ag}$ and $\text{Ta}$ follow the same distribution. In particular,
\[
\text{Ag} \sim \text{Ta} \sim \text{Hypergeometric}(n, |T_u|, k).
\]
By these definitions we have
\begin{equation} \label{eq:E-OPT-replacement-1}
    \sum_{\substack{A \subseteq [n] \\ |A| = k}}\frac{1}{{n\choose {k}}}\sum_{\substack{B \subseteq [n] \\ |B| = k}}\frac{1}{{n\choose{k}}}||B \cap T_u| - |A \cap T_u|| = \mathbb{E}(|\text{Ta} - \text{Ag}|).
\end{equation}
    Using Jensen's inequality, we obtain
\begin{equation} \label{eq:E-OPT-Jensens}
    \mathbb{E}(|\text{Ta} - \text{Ag}|) \leq \sqrt{\mathbb{E}((\text{Ta} - \text{Ag})^2)} = \sqrt{2 \text{Var}(\text{Ag})},
\end{equation}
using the fact $\text{Ag} \sim \text{Ta}$. As $\text{Ag} \sim \text{Hypergeometric}(n, |T_u|, k)$, we have 
\begin{equation} \label{eq:hypergeom-variance}
\text{Var(Ag)} = \frac{k(n-k)}{n-1}\frac{|T_u|}{n}\frac{n -|T_u|}{n},
\end{equation}
so that combining equations (\ref{eq:E-opt-initial}), (\ref{eq:E-OPT-replacement-1}), (\ref{eq:E-OPT-Jensens}) and (\ref{eq:hypergeom-variance}) yields
\begin{align*}
\mathbb{E}(OPT(X_n, k)) &\leq  \sum_{T \in X_n} \text{Pr}(X_n = T)\sum_{u = 1}^{n-1} \sqrt{\frac{k(n-k)}{n-1}}\sqrt{2\frac{|T_u|}{n}\frac{n -|T_u|}{n}}\\
&=  \sqrt{k(n-k)(n-1)} \sum_{T \in X_n} \text{Pr}(X_n = T)\sum_{u = 1}^{n-1} \frac{1}{n-1}\sqrt{2\frac{|T_u|}{n}\frac{n -|T_u|}{n}}\\
&\leq \sqrt{k(n-k)(n-1)} \sqrt{\frac{D}{n-1}}.\\
\end{align*}
Here the last line is obtained by the concavity of $\sqrt{\cdot}$, the fact that $\sum_{T \in X_n}\text{Pr}(X_n = T)\sum_{u = 1}^{n-1} \frac{1}{n-1} = 1$, and Lemma \ref{lem:D-value}. It follows that $\mathbb{E}(OPT(X_n, k)) \leq \sqrt{Dk(n-k)}$.
\end{proof}
When the distribution $X_n$ is uniform, \citet{meirmoon1970} find $D \sim \sqrt{\pi n/2}$. In this case we have
\[
\mathbb{E}(OPT(X_n, k)) \leq \sqrt{k(n-k)\sqrt{\pi n/2}}.
\]
In Figures \ref{fig:n-vs-OPT-expected} and \ref{fig:k-vs-OPT-expected} we plot this bound against the data from Figures \ref{fig:n-vs-OPT} and \ref{fig:k-vs-OPT} respectively.

\begin{figure}[ht]
    \centering
    \includegraphics[width=\textwidth]{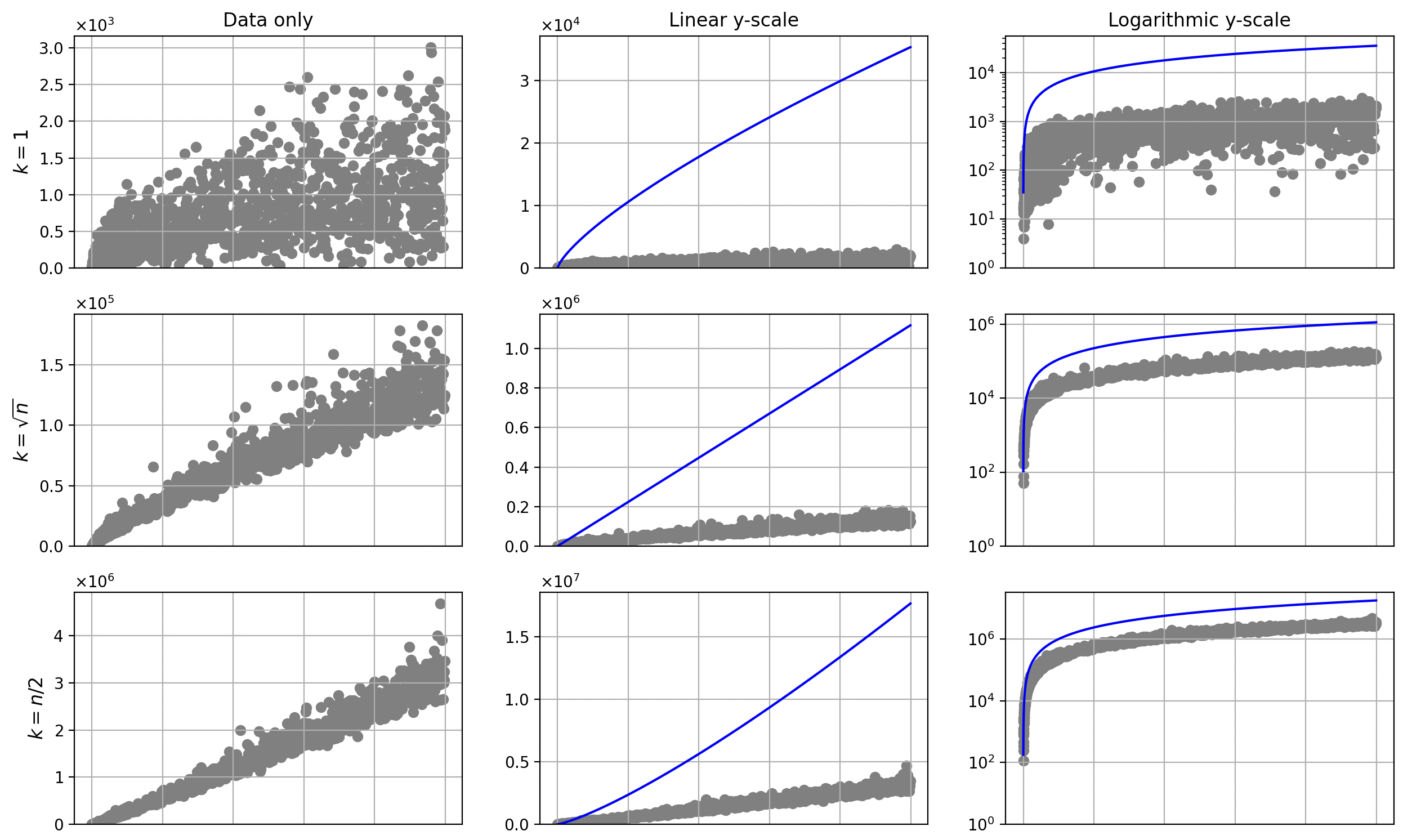}
    \caption{Comparison of $OPT$ (grey dots) and $\sqrt{k(n-k)\sqrt{\pi n/2}}$ (blue line) in random labeled trees with up to $10^6$ nodes.}
    \label{fig:n-vs-OPT-expected}
\end{figure}

\begin{figure}[ht]
    \centering
    \includegraphics[width=\textwidth]{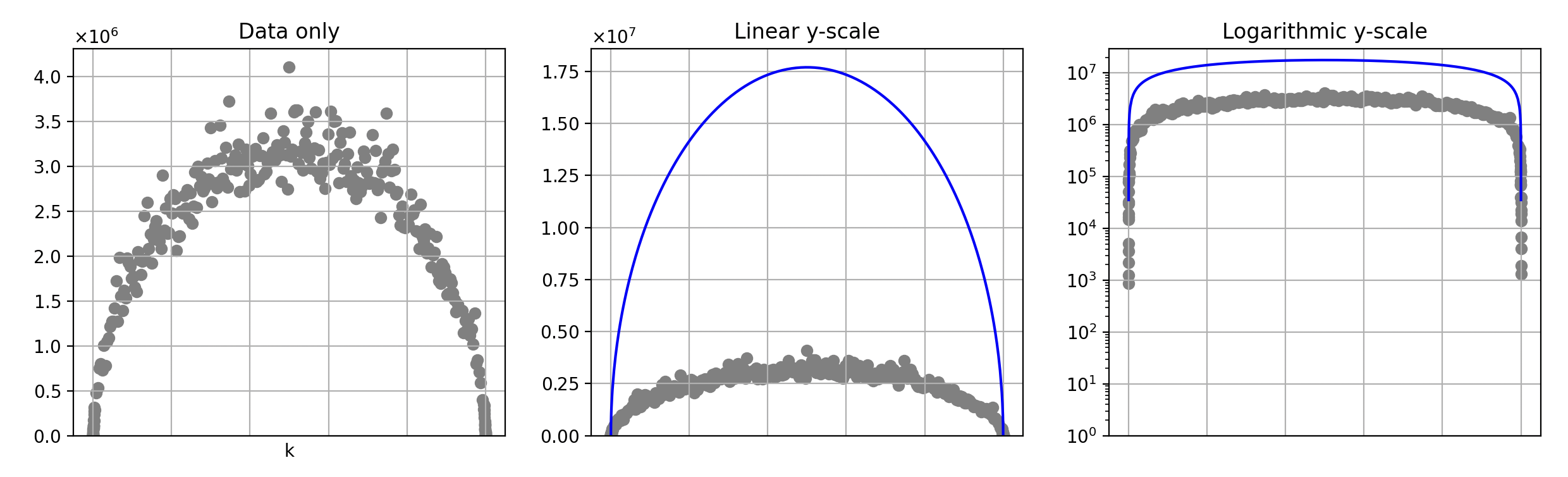}
    \caption{Comparison of $OPT$ (grey dots) and $\sqrt{k(n-k)\sqrt{\pi n/2}}$ (blue line) over $k$ in random labeled trees with $10^6$ nodes.}
    \label{fig:k-vs-OPT-expected}
\end{figure}

This is a significant improvement on the $k(n-k)$ upper bound for estimating solution size, but is still loose. It is particularly loose for fixed small $k$. Indeed, when $k = 1$, we should have $\mathbb{E}(OPT(X_n, k)) = D$ by their respective definitions, but this upper bound gives $\sqrt{D(n-1)}$. Note $D \leq n-1$, but for uniform (and many other) distributions $X_n$ we have $D \ll n-1$.

In the case where $X_n$ is the path on $n$ vertices with probability 1, we have $D = (n-1)/3$ from the mean absolute difference of the uniform distribution. This is a special case where $D = \Theta(n)$, and the bound $D < n-1$ is as close as possible. From Theorem \ref{thm:e-opt-bound} we obtain 
\[
\mathbb{E}(OPT(X_n, k)) \leq \sqrt{k(n-k)(n-1)/3}.
\]
As this bound is loosest for small $k$, and for $k = 1$ we have $\mathbb{E}(OPT(X_n, k)) = \Theta(\sqrt{nk(n-k)})$, we have a heuristic suggesting that
\[
\mathbb{E}(OPT(X_n, k)) = \Theta\left(\sqrt{nk(n-k)}\right)
\]
for all $k$. This is further supported by the data in Figure \ref{fig:path-graph-avg}.

\begin{figure}[ht]
    \centering
    \includegraphics[width=\textwidth]{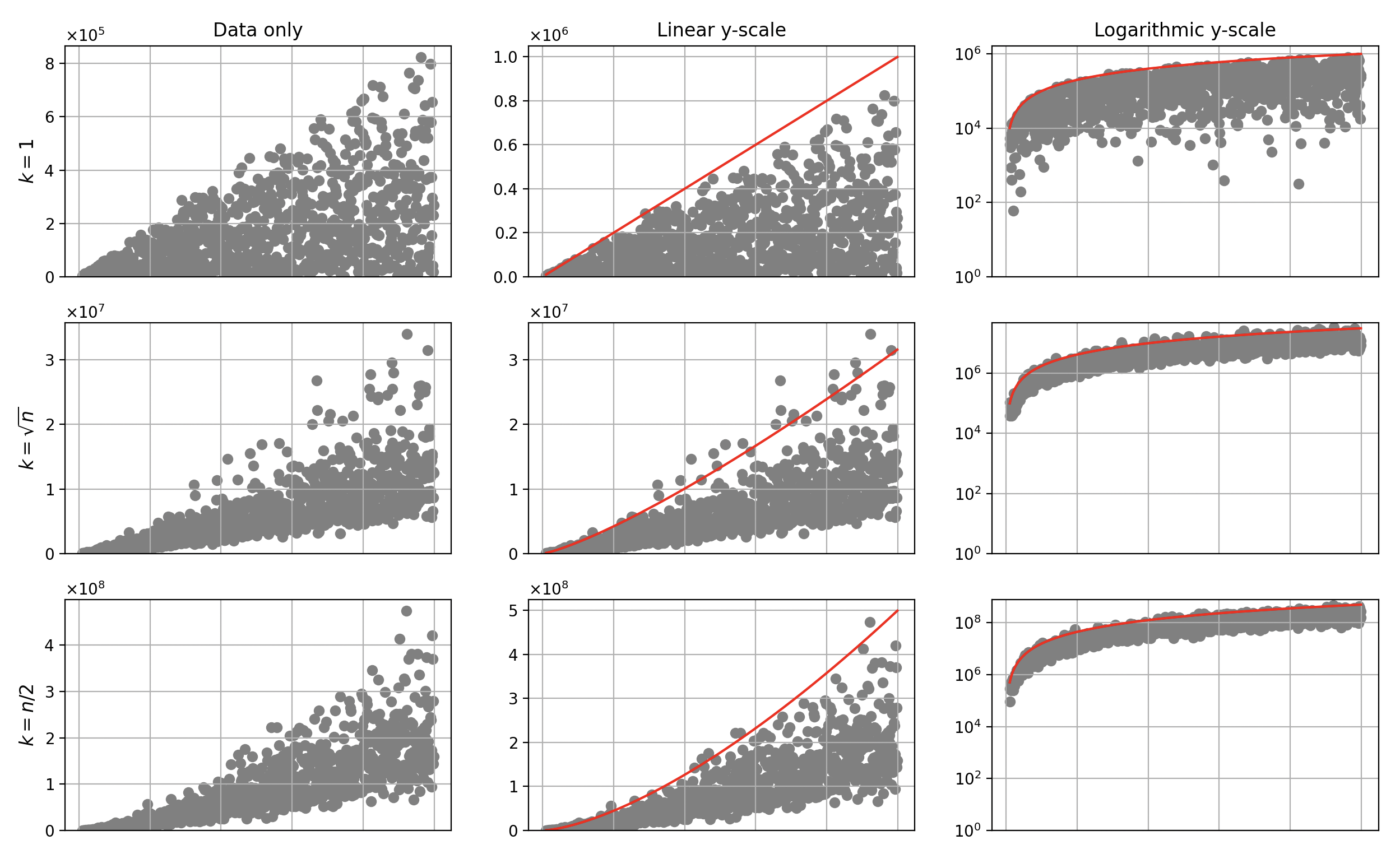}
    \caption{Comparison of $OPT$ (grey dots) and $\sqrt{nk(n-k)}$ (red line) over $k$ in paths with up to $10^6$ nodes.}
    \label{fig:path-graph-avg}
\end{figure}

In particular, if $k$ is some fixed proportion of $n$, the expected number of moves for each individual pebble would be $\Theta(\sqrt{n})$. Generally, for any distribution $X_n$, Theorem \ref{thm:e-opt-bound} gives the bound $O(\sqrt{D})$ for the average distance travelled by an individual pebble.

\section{Conclusions}
\label{sec:conclusion}
Remarkably, the runtime of the UPMT algorithm is equal to the runtime of reading the input and writing the output.
Another interesting property of this algorithm is that for any node $u \in V$ such that $d(u)=0$, UPMT can be solved in $T_u$ independently of the rest of the tree, as no pebble will move from $u$ to its parent (if any). In particular, this means that the for loop at Line \ref{alg:balance_start_for} of \balance{} is ``embarrassingly parallelizable''.

Interestingly, \cref{thm:overall_runtime_UPMT} and \cref{thm:OPT-upper-k-n-k} together show that if plans have maximum length, the worst-case time complexity of our UPMT algorithm is in $O(n^2 \log n)$.
However, \cref{thm:e-opt-bound} on the expected plan length suggests that the average run times would be considerably faster: in $O(n^{1.25} \log n)$ on uniformly distributed trees, for example.

Independent subproblems are also interesting in the context of MAPF, as at least one move per independent subproblem can take place simultaneously without collision. However, we saw that optimal MAPF solutions cannot be found in general by following the demand function.  

Although we do not explicitly give results for general graphs, some of our results carry over by taking any spanning tree of the input graph.
Hence, the upper bound on the plan length (\cref{thm:OPT-upper-k-n-k}) carries over directly, and \cref{thm:overall_runtime_UPMT} provides the time-complexity to find an optimal plan for the chosen spanning tree, which is a feasible plan for the input graph.

Interesting future directions include using this work and so-called ``trans-shipment'' nodes~\citep{ardizzoni2023} to solve Pebble Motion problems on general graphs, as well as automated warehouse planning~\citep[see e.g.][]{kulich2019}. There is also scope for tighter bounds on average plan length, providing improved estimates for plan length in common cases.


\section*{Acknowledgments}

This research is partially supported by the Australian Research Council under the
Discovery Early Career Researcher Award DE240100042.


\bibliographystyle{plainnat}
\bibliography{pebble_tree2}

\end{document}